\documentclass[aps,prx,amssymb,nofootinbib,amsmath,amsfonts,twocolumn,floatfix,superscriptaddress]{revtex4-2}

\usepackage[english]{babel}
\usepackage{amsthm}
\usepackage[toc,page]{appendix}
\usepackage[colorlinks=true,citecolor=blue,linkcolor=magenta]{hyperref}
\usepackage{bbm}
\usepackage{graphicx}
\usepackage{stackengine,xcolor}
\usepackage{mathrsfs}  

\newcommand{\poly}{\operatorname{poly}}

\newcommand{\Tr}{{\rm Tr}}
\newcommand{\Span}{{\rm span}}

\newcommand{\g}{\mathfrak{g}}
\newcommand{\su}{\mathfrak{su}}
\newcommand{\so}{\mathfrak{so}}
\newcommand{\spf}{\mathfrak{sp}}

\renewcommand{\vec}[1]{\boldsymbol{#1}}

\newcommand{\id}{\openone}
\newcommand{\ad}{^{\dagger}}
\newcommand{\ket}[1]{|#1\rangle}
\newcommand{\bra}[1]{\langle #1|}
\newcommand{\ketbra}[2]{|#1\rangle\!\langle #2|}

\newcommand{\twirl}{\TC^{(t)}_{G}\left[ X\right]}
\newcommand{\comm}{\CC^{(t)}_{G}}
\newcommand{\brauer}{\mathfrak{B}_t(-d)}

\newcommand{\BC}{\mathcal{B}}
\newcommand{\CC}{\mathcal{C}}

\newcommand{\GC}{\mathcal{G}}
\newcommand{\HC}{\mathcal{H}}

\newcommand{\LC}{\mathcal{L}}

\newcommand{\NC}{\mathcal{N}}
\newcommand{\OC}{\mathcal{O}}

\newcommand{\SC}{\mathcal{S}}
\newcommand{\TC}{\mathcal{T}}

\newcommand{\SPBB}{\mathbb{SP}}

\newtheorem{theorem}{Theorem}

\newtheorem{corollary}{Corollary}

\newtheorem{proposition}{Proposition}

\begin{document}

\title{Architectures and random properties of symplectic quantum circuits}

\author{Diego Garc\'ia-Mart\'in}
\affiliation{Information Sciences, Los Alamos National Laboratory, 87545 NM, USA}
	
\author{Paolo Braccia}
\affiliation{Theoretical Division, Los Alamos National Laboratory, 87545 NM, USA}
	
\author{M. Cerezo}
\thanks{cerezo@lanl.gov}
\affiliation{Information Sciences, Los Alamos National Laboratory, 87545 NM, USA}

\begin{abstract}
Parametrized and random unitary (or orthogonal) $n$-qubit circuits play a central role in quantum information. As such, one could naturally assume that circuits implementing  symplectic transformations would attract similar attention. However, this is not the case, as $\mathbb{SP}(d/2)$---the group of $d\times d$ unitary symplectic matrices---has thus far been overlooked. In this work, we aim at starting to fill this gap. We begin by presenting  a universal set of generators $\mathcal{G}$ for the symplectic algebra $\mathfrak{sp}(d/2)$, consisting of one- and two-qubit Pauli operators acting on neighboring sites in a one-dimensional lattice. Here, we uncover two critical differences between such set, and equivalent ones for unitary and orthogonal circuits. Namely, we find that the operators  in $\mathcal{G}$  cannot generate arbitrary local symplectic unitaries and that they are not translationally invariant. We then review the Schur-Weyl duality between the symplectic group and the Brauer algebra, and use tools from  Weingarten calculus to prove that Pauli measurements at the output of Haar random symplectic circuits can
converge to Gaussian processes. As a by-product, such analysis  provides us with concentration bounds for Pauli measurements in  circuits that form  $t$-designs over $\mathbb{SP}(d/2)$.  To finish, we present tensor-network tools to analyze shallow random symplectic circuits, and we use these to numerically show that computational-basis measurements  anti-concentrate at logarithmic depth. 
\end{abstract}

\maketitle

\section{Introduction}

The underlying mathematical structures behind the circuits implemented in the standard gate model of quantum computation are those of unitaries and groups. For instance, given an available set of implementable gates one can wonder what kind of interesting evolutions are available by their composition. Here, one can study \textit{specific} combinations of gates (creating a single unitary to solve a given problem), \textit{random} combinations (e.g., average properties as a function of the number of gates taken), or properties of \textit{all} possible combinations (what is the emerging group structure). 
 
The connection between quantum computing and group theory has led to the discovery of  universal gate sets capable of approximating any  evolution in $\mathbb{U}(d)$, the unitary group of dimension $d$~\cite{divincenzo1995two,barenco1995elementary,kitaev1997quantum,kitaev2002classical}. Moreover,  researchers have also studied  architectures that can only implement unitaries from a subgroup of $\mathbb{U}(d)$, such as circuits composed of gates from the Clifford group~\cite{gottesman1998heisenberg,bravyi2021hadamard}, or from some representation of a Lie group, like   matchgate circuits~\cite{de2013power,oszmaniec2017universal,wan2022matchgate,matos2022characterization,cherrat2023quantum,diaz2023parallel,diaz2023showcasing}, group-equivariant circuits~\cite{schatzki2022theoretical,larocca2021diagnosing,monbroussou2023trainability,kerenidis2021classical,jordan2010permutational,zheng2021speeding, zheng2022super} or circuits with translationally-invariant generators~\cite{wiersema2023classification}. The analysis of such architectures has led to insightful results on their classical simulability~\cite{aaronson2004improved,jozsa2008matchgates,anschuetz2022efficient,cerezo2023does,chen2021exponential,aharonov2022quantum,huang2021quantum}, their use in quantum machine learning~\cite{ragone2023unified,fontana2023theadjoint,nguyen2022atheory,larocca2021theory,skolik2022equivariant,meyer2022exploiting,ragone2022representation,larocca2024review}, and on how imposing locality in the generating gates can lead to failures to achieve (subgroup) universality~\cite{zimboras2015symmetry,marvian2022restrictions,marvian2022rotationally,marvian2021qudit,marvian2023non,kazi2023universality}.

In the previous context, the study of random quantum circuits has been particularly active~\cite{fisher2023random}.  These circuits exhibit the appealing feature of being analytically tractable, e.g., via Weingarten calculus~\cite{collins2006integration,collins2022weingarten,mele2023introduction}, providing a test-bed for quantum advantage in sampling problems~\cite{chen2021exponential,boixo2018characterizing,dalzell2021random,oszmaniec2022fermion,arute2019quantum,bouland2019complexity,kondo2022quantum,movassagh2023hardness} and for probing quantum many-body dynamics and the emergence of quantum chaos~\cite{oliveira2007generic,nahum2017quantum,nahum2018operator}. 
For instance, the convergence of random circuits to $t$-designs over $\mathbb{U}(d)$ and the appearance of the anti-concentration phenomenon 
have been the subject of numerous works~\cite{gross2007evenly,dankert2009exact,brandao2016local,harrow2018approximate,hunter2019unitary,haferkamp2023efficient,Haferkamp2022randomquantum,o2023explicit,haah2024efficient,dalzell2022randomquantum,hangleiter2018anticoncentration,braccia2024computing}. Crucially,  
the study of random circuits has been mainly  focused on the unitary group, with significantly less attention being payed to circuits sampled from Lie subgroups of $\mathbb{U}(d)$ (with some notable recent exceptions~\cite{oszmaniec2022fermion,haah2024efficient,braccia2024computing}).

In this work, we contribute to the body of knowledge of circuits that belong to subgroups of $\mathbb{U}(d)$ by studying quantum circuits implementing transformations from  the compact symplectic group $\SPBB(d/2)$ (see Fig.~\ref{fig:schematic}). This Lie group consists of all the $d\times d$ unitary symplectic matrices, which are unitaries that preserve a non-degenerate anti-symmetric bilinear form $\Omega$. Despite its importance in random matrix theory~\cite{mehta2004random}, and classical~\cite{goldstein2001classical} and quantum~\cite{ferraro2005gaussian} dynamics, this group has been mostly neglected in the recent literature. Let us recall that the compact symplectic group is the \emph{only} proper subgroup of the special unitary group $\mathbb{SU}(d)$ that exhibits pure-state controllability~\cite{schirmer2002criteria}. This means that any pure quantum state can be reached from any other pure quantum state by means of a symplectic unitary. This fact, together with the observation that the symplectic group is the key group describing the phase-space dynamics of both classical and quantum systems provides the necessary motivation to study symplectic quantum circuits (as they could be used to simulate such dynamics on a quantum computer).

We begin by discussing how the non-uniqueness of  $\Omega$ is a salient and important feature of $\SPBB(d/2)$ that is not present when studying circuits that implement evolutions from the unitary or orthogonal groups. This up-to-congruence freedom can be exploited to show that when the canonical form of $\Omega$ is used, we can find a  set of generators for the Lie algebra $\mathfrak{sp}(d/2)$ consisting of one- and two-qubit Paulis acting on neighboring sites in a one-dimensional lattice. This set of generators leads to quantum circuit architectures that implement symplectic transformations and that are universal in $\SPBB(d/2)$. Remarkably, these circuits cannot be built from translationally-invariant local generators~\cite{wiersema2023classification}. In fact, circuits built from locally-symplectic quantum gates do not necessarily produce globally-symplectic transformations, but instead span the entire special unitary group $\mathbb{SU}(d)$.  

After identifying how to produce symplectic evolutions on a quantum computer, we review the Schur-Weyl duality between the symplectic group and the Brauer algebra, showing that it can be used along with the Weingarten calculus~\cite{collins2006integration,collins2022weingarten} (which we present via tensor notation) to compute average properties of symplectic random circuits. In particular, we prove that the outputs of Haar random symplectic  circuits can converge in distribution to Gaussian Processes (GPs) when the measurement operator  is  traceless and involutory. This fact allows us to provide concentration bounds for Haar random symplectic circuits, and show that Pauli expectation values concentrate exponentially in the Hilbert space dimension. That is, doubly exponentially in the number of qubits. Furthermore, we give concentration bounds for random circuits that form $t$-designs over $\SPBB(d/2)$.  Finally, following the results in Ref.~\cite{braccia2024computing}, we  present tensor-network-based tools capable of analyzing average properties of shallow symplectic random circuits. Notably, we use these to numerically show  that computational-basis measurement appear to anti-concentrate at logarithmic-depth, indicating that these circuits may be used in quantum supremacy experiments~\cite{chen2021exponential,boixo2018characterizing,dalzell2021random,oszmaniec2022fermion,arute2019quantum,bouland2019complexity,kondo2022quantum,movassagh2023hardness}.

\section{Preliminaries}
\label{sec:preliminaries}

In this section, we introduce some  basic concepts that will be used throughout this work.
We begin by recalling that the standard representation of the compact symplectic group $\mathbb{SP}(d/2):=\mathbb{SP}(d;\mathbb{C})\cap \mathbb{SU}(d)$  consists of all $d\times d$ unitary matrices (with $d$ an even number), such that any $ S\in \mathbb{SP}(d/2)$  satisfies the relation
\begin{equation}\label{eq:symplectic-unitaries}
    S^{T} \Omega S = \Omega\,,
\end{equation}
where $\Omega$ is a  non-degenerate anti-symmetric bilinear form. In other words, $\mathbb{SP}(d/2)$ is the group of unitary matrices that preserve the product  $\vec{x}^T \Omega \,\vec{y}$ for vectors $\vec{x},\vec{y}\in\mathbb{C}^{d}$. Then, we recall that the Lie algebra associated with $\mathbb{SP}(d/2)$ is the symplectic Lie algebra,  denoted as $\mathfrak{sp}(d/2)$, whose  elements are $d\times d$ anti-Hermitian matrices, such that any $M\in \mathfrak{sp}(d/2)$  satisfies 
\begin{equation}\label{eq:symplectic-algebra}
    M^T\Omega=-\Omega M\,.
\end{equation}
Moreover, any orthogonal basis for  the symplectic Lie algebra is of dimension $d(d+1)/2$.

Here, we remark that $\Omega$ in Eqs.~\eqref{eq:symplectic-unitaries} and~\eqref{eq:symplectic-algebra}  is not uniquely defined. Typically, one uses  the Darboux basis---or canonical form-- in which  $\Omega$  takes the form 
\begin{equation} \label{eq:omega}
	\Omega=\begin{pmatrix} 0& \id_{d/2} \\ - \id_{d/2} & 0\end{pmatrix}\,,
\end{equation}
with $\id_{d/2}$ being the $d/2 \times d/2$ identity matrix. In this work we will assume that $\Omega$ is given by Eq.~\eqref{eq:omega}, as any other non-degenerate anti-symmetric bilinear form $\Omega'$ can always be mapped to $\Omega$ by a  change of basis $Q\, \Omega' Q^T=\Omega$, with a $d\times d$ orthogonal matrix $Q$ (i.e., such that $Q^T=Q^{-1}$). To finish, we recall that $\Omega$ has the elementary properties 
\begin{equation}
    \Omega^2=-\id_d\,,\quad \text{and} \quad\Omega\Omega^T=\Omega^T\Omega=\id_d\,.
\end{equation}

\begin{figure}[t]
    \centering
    \includegraphics[width=1\linewidth]{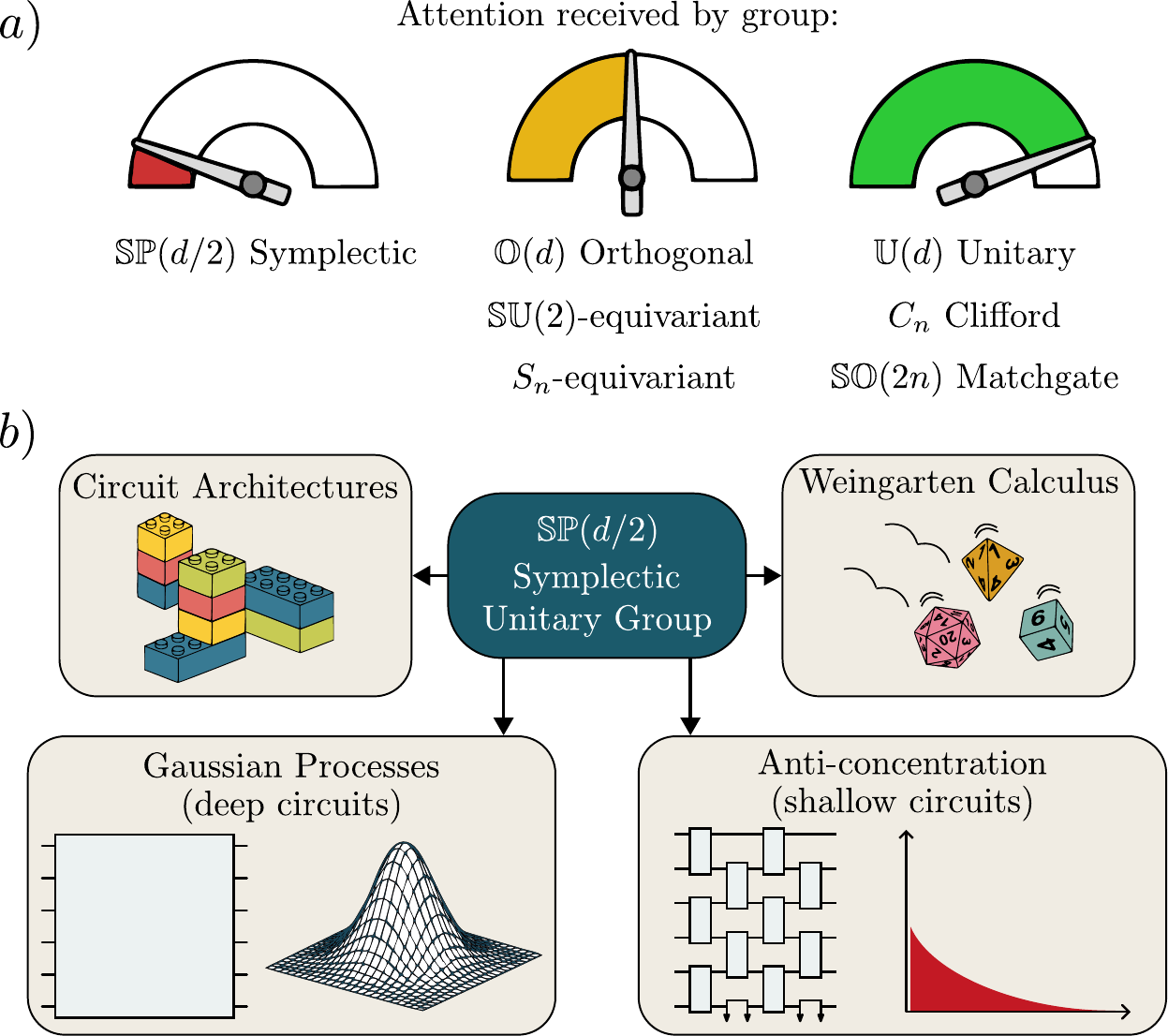}
    \caption{{\bf Schematic representation of our main results}. a) When compared against other groups, the compact  group $\mathbb{SP}(d/2)$ of $d\times d$ symplectic unitaries has received considerably less attention. b) Here we introduce tools to study $\mathbb{SP}(d/2)$, such as  presenting easy-to-implement circuit architectures capable of producing any symplectic evolution. We also review the Weingarten calculus for this group and use it to study properties of random symplectic circuits, like their convergence to Gaussian processes (deep circuits) or the emergence of anti-concentration (shallow circuits).    }
    \label{fig:schematic}
\end{figure}

\section{Pauli operator basis for the  symplectic Lie algebra}

Let us now  focus on the case when $d=2^n$ so that the symplectic unitaries act on the Hilbert space $\HC=(\mathbb{C}^2)^{\otimes n}$ of $n$ qubits. With this choice one can verify that
\begin{equation}\label{eq:omega-qubit}
    \Omega=i Y\otimes\id^{\otimes n-1}\,, 
\end{equation}
with $Y$ the Pauli matrix and $\id$ the $2\times 2$ identity. Here, we ask the following question: \textit{What is a natural choice for the basis elements of the standard representation  of the symplectic Lie algebra $\mathfrak{sp}(d/2)$?} As we prove in Appendix~\ref{ap:prop-1}, the following proposition holds.

\begin{proposition} \label{prop:sp-algebra}
     A basis for the standard representation of the $\mathfrak{sp}(d/2)$ algebra is 
     \begin{equation} \label{eq:sp-dla}
	   B_{\mathfrak{sp}(d/2)}\equiv i\{\{X,Y,Z\}\otimes P_s \}\,\cup\, i\{\id\otimes P_a\}\,,
    \end{equation}
    where $P_s$ and $P_a$ belong to the sets of arbitrary symmetric and anti-symmetric Pauli strings on $n-1$ qubits, respectively, and $\id,X,Y,Z$ are the usual $2\times 2$ Pauli matrices.
 \end{proposition}

We recall that $\{P_s\}$ and $\{P_a\}$ are composed of all Paulis acting on $n-1$ qubits with an even or odd number of $Y$'s, respectively. It is interesting to note that Eqs.~\eqref{eq:omega-qubit} and~\eqref{eq:sp-dla} reveal that the first qubit plays a privileged role. As we will see below,  this asymmetry will translate into the structure of symplectic quantum circuits. In particular, it will be responsible for the lack of translational invariance in the generators of the circuit.

\begin{figure}[t]
    \centering
    \includegraphics[width=.9\linewidth]{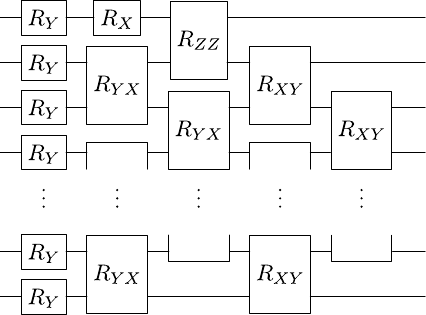}
    \caption{{\bf Quantum circuits for symplectic unitaries}. Example of the basic building block for the implementation of symplectic unitary transformations on a quantum computer. The notation $R_{H_l}$ stands for $e^{i\theta_l H_l}$, with independent $\theta_l$ angles in each gate. As stated in Theorem~\ref{th:symplectic-universal}, the Lie closure of the generators appearing in this circuit, which are not translationally invariant, produces   $\spf(d/2)$. This implies that any symplectic unitary from the $\SPBB(d/2)$ group can be implemented by a quantum circuit architecture consisting of blocks of this form.}
    \label{fig:global-sp-circuits}
\end{figure}

\section{quantum circuits for symplectic unitaries}

The fact that the matrices in $\mathbb{SP}(d/2)$ are unitary implies that they can be implemented by quantum circuits. While some architectures for such symplectic unitaries have been found~\cite{schirmer2002identification,zeier2011symmetry}, they do not make use of the canonical form of $\Omega$ and are composed of non-local gates obtained by either correlating parameters~\cite{zeier2011symmetry} or by using non-local generators~\cite{schirmer2002identification}. 

Our first contribution is to show that by taking $\Omega$ as in Eq.~\eqref{eq:omega-qubit}, we can find a set local generators $\GC$ for which circuits of the form 
\begin{equation} \label{eq:circuit}
    U= \prod_l e^{ i \theta_l H_l}\,,
\end{equation}
where $\theta_l$ are real-valued parameters and  $H_l\in \GC$, are universal and can therefore produce any unitary in $\mathbb{SP}(d/2)$. In particular, the following theorem, whose proof can be found in Appendix~\ref{ap:th-1}, holds. 
\begin{theorem} \label{th:symplectic-universal}
    The set of unitaries of the form in Eq.~\eqref{eq:circuit}, with  generators taken from
    \begin{equation}\label{eq:generators}
        \GC = \{Y_i\}_{i=1}^{n} \cup \{X_iY_{i+1}, Y_iX_{i+1}\}_{i=2}^{n-1} \cup X_1 \cup Z_1 Z_2\,,
    \end{equation} 
is universal in $\mathbb{SP}(d/2)$, as
    \begin{equation}
    \Span_{\mathbb{R}} \langle i \GC\rangle_{\rm Lie}=  \mathfrak{sp}(d/2)\,.
\end{equation}
Here, $\langle i\GC\rangle_{\rm Lie}$ is the Lie closure of $i\GC$, i.e., the set of operators obtained by the nested commutation of the elements in $i\GC$.
\end{theorem}    
\noindent In Eq.~\eqref{eq:generators}, $X_i$, $Y_i$ and $Z_i$ denote the Pauli operators acting on the $i$-th qubit. 

Let us now discuss the implications of Theorem~\ref{th:symplectic-universal}. First, we note that the quantum circuits obtained from the set of generators in Eq.~\eqref{eq:generators} can be implemented with  one- and two-qubit  gates acting on nearest neighbors on a one-dimensional chain of qubits with open boundary conditions (see  Fig.~\ref{fig:global-sp-circuits}). Moreover,  each gate has an independent parameter. These features render the circuits readily implementable with the topologies and connectivities available in near-term quantum hardware. 

A second important implication of  Theorem~\ref{th:symplectic-universal} is that symplectic circuits are not translationally invariant in the sense that the local generators  $\{H_l\}$ are not the same on each pair of adjacent qubits.  This is in stark contrast with the unitary and orthogonal groups $\mathbb{U}(d)$ and $\mathbb{O}(d)$, as these can be constructed from translationally invariant generators~\cite{wiersema2023classification}. As mentioned in the previous section, the lack of translational invariance for the symplectic group can be traced back to the asymmetric structure of the $\Omega=iY\otimes\id^{\otimes n-1}$ matrix. 

In fact, we can see that to construct quantum circuits that implement  symplectic transformations in $\mathbb{SP}(d/2)$, one can choose local generators from the special orthogonal algebra $\mathfrak{so}(4)$  acting on the last $n-1$ qubits, such as $\{Y\otimes \id, \id\otimes Y, X\otimes Y, Y\otimes X\}$. The reason is that  $i\id\otimes P_a$ belongs to $B_{\spf(d/2)}$ for all anti-symmetric Paulis $P_a$ according to Eq.~\eqref{eq:sp-dla}, and this set is a basis for $\so(d/2)$. 
Then, in order to generate $\so(d/2)$ in the last $n-1$ qubits, it suffices to employ the local generators of $\so(4)$ on each pair of nearest neighbors in those $n-1$ qubits~\cite{wiersema2023classification}. Following this reasoning, we now need a different set of generators acting on the first pair of qubits.  It can be shown that adding operators from the $\spf(2)$ algebra such as $\{X\otimes \id,Y\otimes \id,Z\otimes Z\}$ acting on the first pair of qubits completely generates  $\spf(d/2)$, and nothing else (see Appendix~\ref{ap:th-1}).

To finish, we note that we have defined symplectic transformations with respect to the canonical form of the $\Omega$ matrix in Eq.~\eqref{eq:omega-qubit}. If we were to choose a different $\Omega$, there always exist an orthogonal change of basis that would take as back to the canonical form, as explained in Sec.~\ref{sec:preliminaries}. This would then correspond to global unitaries acting at the beginning and end of the circuit.

\begin{figure}[t]
    \centering
    \includegraphics[width=.9\linewidth]{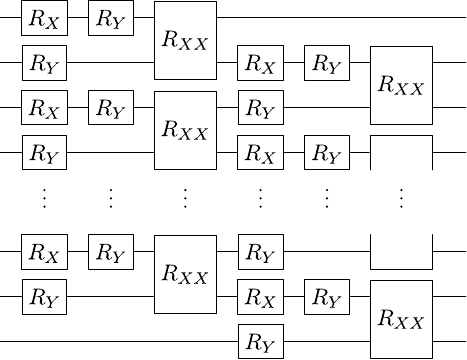}
    \caption{{\bf Circuits with local symplectic gates are not symplectic}. Example of a building block consisting of two-qubit symplectic gates.  This architecture is universal for quantum computation, as the Lie closure of all the generators leads to $\su(d)$, and thus any (special) unitary transformation can be decomposed into a circuit consisting of these gates.}
    \label{fig:local-sp-circuits}
\end{figure}

\section{Circuits with local symplectic gates are not symplectic}
\label{sec:local-symp-circs}

In the previous section we have shown that one can generate globally symplectic unitaries in $\mathbb{SP}(d/2)$ by implementing locally symplectic unitaries on the first two qubits, plus orthogonal unitaries acting on the second through last qubits. This raises the question as to what happens if we construct a circuit where all gates are locally symplectic (including those acting on the second through last qubits). For example, we can consider circuits such as those in Fig.~\ref{fig:local-sp-circuits}, where local gates from $\mathbb{SP}(2)$ are implemented on neighboring qubits on a  one-dimensional connectivity. 

Following Proposition~\ref{prop:sp-algebra}, one possible choice of  generators for the local gates that produce universal $\mathbb{SP}(2)$ circuits is $\{X\otimes \id ,Y\otimes \id, \id\otimes Y, X\otimes X\}$, since these suffice to generate all the basis $B_{\spf(2)}$ elements via (nested) commutation. Given that this set of  generators is  translationally invariant, it falls under the general classification of Ref.~\cite{wiersema2023classification}.
In particular, it is known that they produce unitary universal circuits (up to a global phase), that is, their Lie closure leads to  $\su(d)$. For the shake of completeness, we formalize this claim in the following proposition, proved in Appendix~\ref{ap:prop-2}.

\begin{proposition} \label{prop:symplectic-local}
    The set of unitaries of the form in Eq.~\eqref{eq:circuit}, with  generators taken from
    \begin{equation}\label{eq:generators-local}
        \GC_L = \bigcup_{i=1}^{n-1}\{X_i ,Y_i, Y_{i+1}, X_iX_{i+1}\}\,,
    \end{equation} 
is universal in $\mathbb{SU}(d)$, as
    \begin{equation}
    \Span_{\mathbb{R}} \langle i \GC_L\rangle_{\rm Lie}=  \mathfrak{su}(d)\,.
\end{equation}
\end{proposition}

Ultimately, the expressive power of locally-symplectic circuits stems from the simple observation that the compact symplectic group is not amenable to the tensor product structure of the Hilbert space of qubits, in contrast to the orthogonal and unitary groups. More precisely, let $U,O$ be  unitary and  orthogonal matrices, respectively. Then, $\id_s \otimes U$ is unitary for any Hilbert space partition index $s$, and analogously for $O$. However, if $S$ is a  symplectic matrix, then $\id_s \otimes S$ is not symplectic in general. Indeed, taking $\Omega$ from Eq.~\eqref{eq:omega-qubit}, we have that  $\id\otimes S^T\, \Omega\, \id\otimes S=iY\otimes S^T S$, which is not equal to $\Omega$ unless $S$ is also orthogonal. This implies that local symplectic generators that do not belong to the special orthogonal Lie algebra (e.g.,  $X\otimes \id$ and $X\otimes X$) on the last $n-1$ qubits  are no longer in the symplectic algebra when tensored with identities on the rest of the qubits. Hence, it is clear that quantum circuits with locally-symplectic gates will be able to generate non-symplectic transformations.

\section{Symplectic Weingarten calculus}
\label{sec:weingarten}

\begin{figure*}[t]
    \centering
    \includegraphics[width=\linewidth]{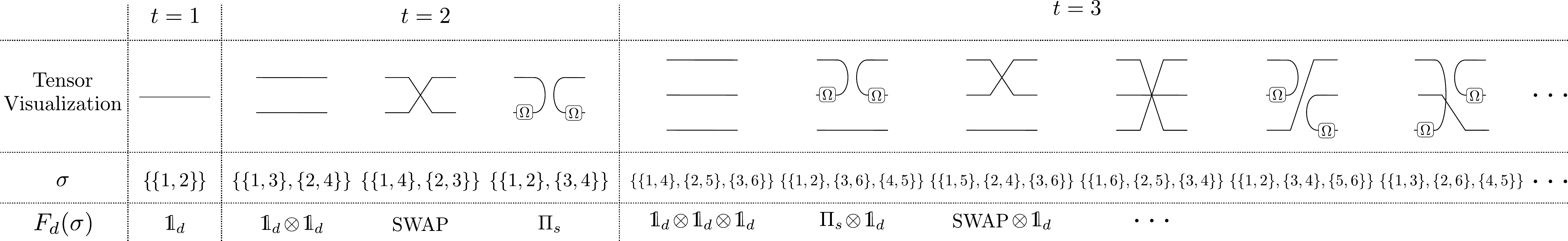}
    \caption{{\bf Elements of the Brauer algebra $\mathfrak{B}_t(-d)$.} Here we present all the elements of $\mathfrak{B}_t(-d)$ for $t=1$ and $t=2$, as well as some elements for $t=3$. In all cases we introduce their tensor representation visualization, their decomposition into disjoint pairs, and when convenient, their representation $F_d$. We use the convention whereby we first enumerate all the items on the left-hand side from top to bottom and then the items on the right-hand side (also from top to bottom). }
    \label{fig:brauer}
\end{figure*}

Now that we know how to construct quantum circuits that implement symplectic transformations, we turn to study their average properties. In particular, if we assume that the circuits that we implement sample unitaries according to the Haar measure over $\mathbb{SP}(d/2)$,  either exactly or approximately,  we can leverage the tools from the symplectic Weingarten calculus~\cite{collins2006integration,collins2022weingarten}. For ease of notation, we will also  use diagrammatic tensor notation to simplify computations. We refer the reader to Refs.~\cite{mele2023introduction,garcia2023deep} for an in-depth treatment of Weingarten calculus on the unitary and orthogonal groups from a quantum information perspective.

The goal of Weingarten calculus is to compute integrals of polynomials in the entries of matrices (and their complex conjugates) over the left-and-right-invariant Haar measure on a compact matrix Lie group $G$. This can be shown to be equivalent to computing matrix entries  of the following operator,
\begin{equation} \label{eq:twirl}
    \TC^{(t)}_{G}\left[ X\right] =\int_{G} d\mu(U)\, U^{\otimes t} X (U^\dagger)^{\otimes t} \,.
\end{equation}
Here,  $\TC^{(t)}_{G}\left[ X\right]$ is called the $t$-th fold twirl of $X$ over  $G$, $X$  belongs to the set of bounded operators $\BC(\HC^{\otimes t})$ acting on $\HC^{\otimes t}$, and $d\mu(U)$ is the Haar measure on $G$. Importantly, in practice one often encounters circuits that are not fully Haar random but that are sufficiently so to reproduce the first $t$ moments of the Haar random distribution, i.e., to match the twirl in Eq.~\eqref{eq:twirl}. These are called $t$-designs.

  It is straightforward to show that the twirl is an orthogonal projector onto the $t$-th order commutant $\comm$ of the tensor representation of $G$, that is, the vector subspace of all matrices that commute with $U^{\otimes t}$ for all $U\in G$. 
Hence, we can write
\begin{equation} \label{eq:twirl-expansion}
    \TC^{(t)}_{G}\left[ X\right] =\sum_{\mu,\nu} W^{-1}_{\mu\nu}\Tr[P_\nu\ad X]  P_\mu\,,
\end{equation}
where the $\{P_\mu\}$ operators are a basis that spans $\comm$ (note that they need not be orthonormal, nor Hermitian), and $W$  is the  Gram matrix of the aforementioned basis with respect to the Hilbert-Schmidt inner product, i.e., it is the matrix whose entries are $W_{\mu\nu}=\Tr[P_\mu^\dagger P_\nu]$. 
In summary, in order to compute $\twirl$, one needs to find a set of operators spanning $\comm$, compute the corresponding Gram matrix, and invert it (or in some cases perform the pseudo-inverse).

Perhaps the main ingredient necessary for using Eq.~\eqref{eq:twirl-expansion}, is the knowledge of a basis for $\comm$. While in some cases such basis might not be readily available, when $G$ is the standard representation of a unitary, orthogonal or symplectic group,  one can use 
the Schur-Weyl duality to obtain such basis. In particular, when $G$ is the  unitary group, then $\comm$  is found to be spanned by a representation of the symmetric group $S_t$~\cite{harrow2024approximate}, whereas if $G$ is the orthogonal or the symplectic group, its commutant is spanned by some representation of the Brauer algebra $\mathfrak{B}_t(\delta)$~\cite{collins2006integration} (with $\delta=d$ for the orthogonal group, and $\delta=-d$ for the symplectic group).

We recall that the Brauer algebra $\mathfrak{B}_t(\delta)$ is an associative algebra that has a basis consisting of all possible pairings of a set of size $2t$. That is, given a set of $2t$ items, the basis elements of the Brauer algebra correspond to all possible ways of splitting them into pairs. This has two important implications. First, we can see that all the permutations  in $S_t$ are also in $\mathfrak{B}_t(\delta)$, as these corresponds to the pairings that can only connect the first $t$ items to the remainder ones. Second, a straightforward calculation reveals that there are  $D_t=\frac{(2t)!}{2^t t!}=(2t-1)!!$  elements in the aforementioned basis of the Brauer algebra. Here we also note that  every basis element $\sigma\in\mathfrak{B}_t(\delta)$ can  be completely specified by $t$ disjoint pairs, as 
\begin{equation}
    \sigma=\{\{\lambda_1, \sigma(\lambda_1)\}\cup\dots\cup\{\lambda_t, \sigma(\lambda_t)\}\}\,.
\end{equation}
In Fig.~\ref{fig:brauer}, we diagrammatically show all the elements $\sigma\in\mathfrak{B}_t(-d)$ for $t=1,2$ (as well as some for $t=3$) using tensor representation. 
Additionally, a Brauer algebra $\mathfrak{B}_t(\delta)$ depends on a parameter $\delta$ and has the structure of a $\mathbb{Z}(\delta)$-algebra. This implies that when we multiply two basis elements in $\mathfrak{B}_t(\delta)$, we do not necessarily obtain a basis element from $\mathfrak{B}_t(\delta)$ but rather a basis element in $\mathfrak{B}_t(\delta)$ times an integer power of $\delta$. Diagrammatically, this means that when we connect (multiply) two diagrams, closed loops can appear. Then, the power to which the factor $\delta$ is raised is equal to the number of closed loops formed.

While the previous determines how the  abstract Brauer algebra is defined, we still need to specify how its elements are represented and how they act on $\HC^{\otimes t}$. In particular, we here consider the representation $F_d:\mathfrak{B}_t(-d)\rightarrow \BC(\HC^{\otimes t})$ such that 

\small
\begin{align} 
F_d(\sigma) = \sum_{i_1,\dots,i_{2t}=1}^{d}\prod_{\gamma=1}^{t}  &\Omega_{{\sigma(\lambda_\gamma)}}^{h(\lambda_\gamma,\sigma(\lambda_\gamma))}\ket{i_{t+1},i_{t+2},\dots,i_{2t}} \label{eq:rep-b_t}\\ &\times\bra{i_1,i_2,\dots,i_t} \,\Omega_{{\sigma(\lambda_\gamma)}}^{h(\lambda_\gamma,\sigma(\lambda_\gamma))}\delta_{i_{\lambda_\gamma}, i_{\sigma(\lambda_\gamma)}} \,,\nonumber
\end{align} 
\normalsize
where $h(\lambda_\gamma,\sigma(\lambda_\gamma))=1$ if $\lambda_\gamma,\sigma(\lambda_\gamma)\leq n$ or if $\lambda_\gamma,\sigma(\lambda_\gamma)> n$ and zero otherwise, and where $\Omega_{\sigma(\lambda_\gamma)}$ indicates that the $\Omega$ matrix acts on the  $\sigma(\lambda_\gamma)$-th copy of the Hilbert space.

Equipped with the previous knowledge, let us consider specific values of $t$. In each case, we will present the basis elements of $\mathfrak{B}_t(-d)$ as well as explicitly  compute the formula for the  twirl in Eq.~\eqref{eq:twirl-expansion}. First, we consider the case when $t=1$. As shown in Fig.~\ref{fig:brauer},  $\mathfrak{B}_1(-d)$ contains a single element $\{\{1,2\}\}$ whose representation is given by 
\begin{equation}
    F_d(\{\{1,2\}\})=\sum_{i_1=1}^d\ket{i_1}\bra{i_1}:=\id_d\,,
\end{equation}
which indeed confirms that the representation of $\mathbb{SP}(d/2)$ is irreducible (the only element in the $t=1$ commutant is the identity).  As such, we find
\begin{equation}
    W=\begin{pmatrix}
    d
    \end{pmatrix}\,,
\end{equation}
 and thus
\begin{align}\label{eq:twirl-1}
    \TC^{(1)}_{\mathbb{SP}(d/2)}[X]=\frac{\Tr[X]}{d}\,\id_d\,.
\end{align}
Then,  when $t=2$, $\mathfrak{B}_2(-d)$ contains three elements given by  $\{\{1,3\},\{2,4\}\}$, $\{\{1,4\},\{2,3\}\}$, and $\{\{1,2\},\{3,4\}\}$, whose representations are 
\small
    \begin{align}
    F_d(\{\{1,3\},\{2,4\}\})&=\sum_{i_1,i_2=1}^d\ket{i_1i_2}\bra{i_1i_2}:=\id_d\otimes \id_d\,,\nonumber\\
    F_d(\{\{1,4\},\{2,3\}\})&=\sum_{i_1,i_2=1}^d\ket{i_1i_2}\bra{i_2i_1}:={\rm SWAP}\,,\nonumber\\
    F_d(\{\{1,2\},\{3,4\}\})&=\sum_{i_1,i_2=1}^d\id_d\otimes\Omega\ket{i_1i_1}\bra{i_2i_2}\id_d\otimes\Omega:=\Pi_s\,.\nonumber
\end{align}
\normalsize
Recalling that the maximally-entangled Bell state $\ket{\Phi^+}$ between the two copies of $\HC$ is $\ket{\Phi^+}=\frac{1}{\sqrt{d}}\sum_{i_1=1}^{d}\ket{i,i}$, we find that $    \Pi_s= d (\id_d \otimes \Omega)  \,\ketbra{\Phi^+}{\Phi^+}(\id_d\otimes \Omega)$. The identification of $\Pi_s$ with the Bell state shows that $\Pi_s$ satisfies an analogous of the so-called ricochet property~\cite{nielsen2000quantum},
\begin{equation} \label{eq:ricochet} \begin{split}
    (A\otimes B)\Pi_s&=d \,\left(\id_d\otimes B\Omega A^T \right)\ketbra{\Phi^+}{\Phi^+}\,\id_d \otimes \Omega\\&=d \,\left(A\Omega^T B^T\otimes \id_d\right)\ketbra{\Phi^+}{\Phi^+}\, \id_d \otimes \Omega\,,
\end{split}
\end{equation}
which allows one to readily verify that $\Pi_s$ belongs to $\CC_{\SPBB(d/2)}^{(2)}$.
In this $t=2$ case,  the Gram matrix is
\begin{equation} \label{eq:gram-t2}
    W= \begin{pmatrix}
        d^2& d & -d \\ d & d^2 & d \\ -d & d & d^2   
    \end{pmatrix}\,,
\end{equation}
leading to the formula for the two-fold twirl, 
\footnotesize
\begin{align}  
    \TC^{(2)}_{\mathbb{SP}(d/2)}\left[X\right]&=\frac{(d-1)\Tr[X]  -\Tr[X{\rm SWAP}]  +\Tr[X \Pi_s]}{d(d+1)(d-2)}\,\id_d\otimes \id_d\nonumber\\
    &+\frac{-\Tr[X] + (d-1)\Tr[X{\rm SWAP}]  -\Tr[X \Pi_s]}{d(d+1)(d-2)} \,{\rm SWAP}\nonumber\\
    &+\frac{\Tr[X]  -\Tr[X{\rm SWAP}] + (d-1) \Tr[X \Pi_s]}{d(d+1)(d-2)}\,\Pi_s\,.\label{eq:twirl-2}
\end{align}
\normalsize
We refer to Fig.~\ref{fig:gram} for a visualization in tensor notation of how the elements of the Gram matrix~\eqref{eq:gram-t2} are computed.

Given that the dimension of $\mathfrak{B}_t(-d)$ is $D_t=(2t-1)!!$, keeping track of all the elements in the commutant of  $\SPBB(d/2)$ quickly becomes intractable as  $t$ grows.  However, we can derive asymptotic formulas that can be used to perform calculations in the large $d$ limit.  Here,  the key is to realize that the Gram matrix is given by
\begin{equation}\label{eq:gram-split}
    W = d^t\left(\id_{D_t}+\frac{1}{d} B\right)\,,
\end{equation}
with  $B$ a matrix whose  entries are in $\OC(1)$. We refer the reader to Appendix~\ref{ap:asympt} for additional details on why $W$ takes this form, but it should be noted that  there is a very close analogy between the approximate orthogonality of permutations in the $t$-fold commutant of the unitary group (see~\cite{harrow2024approximate} and~\cite{garcia2023deep}) and that of the generators of  the Brauer algebras in the orthogonal's and symplectic's commutants.

\begin{figure}[t]
    \centering    \includegraphics[width=0.9\columnwidth]{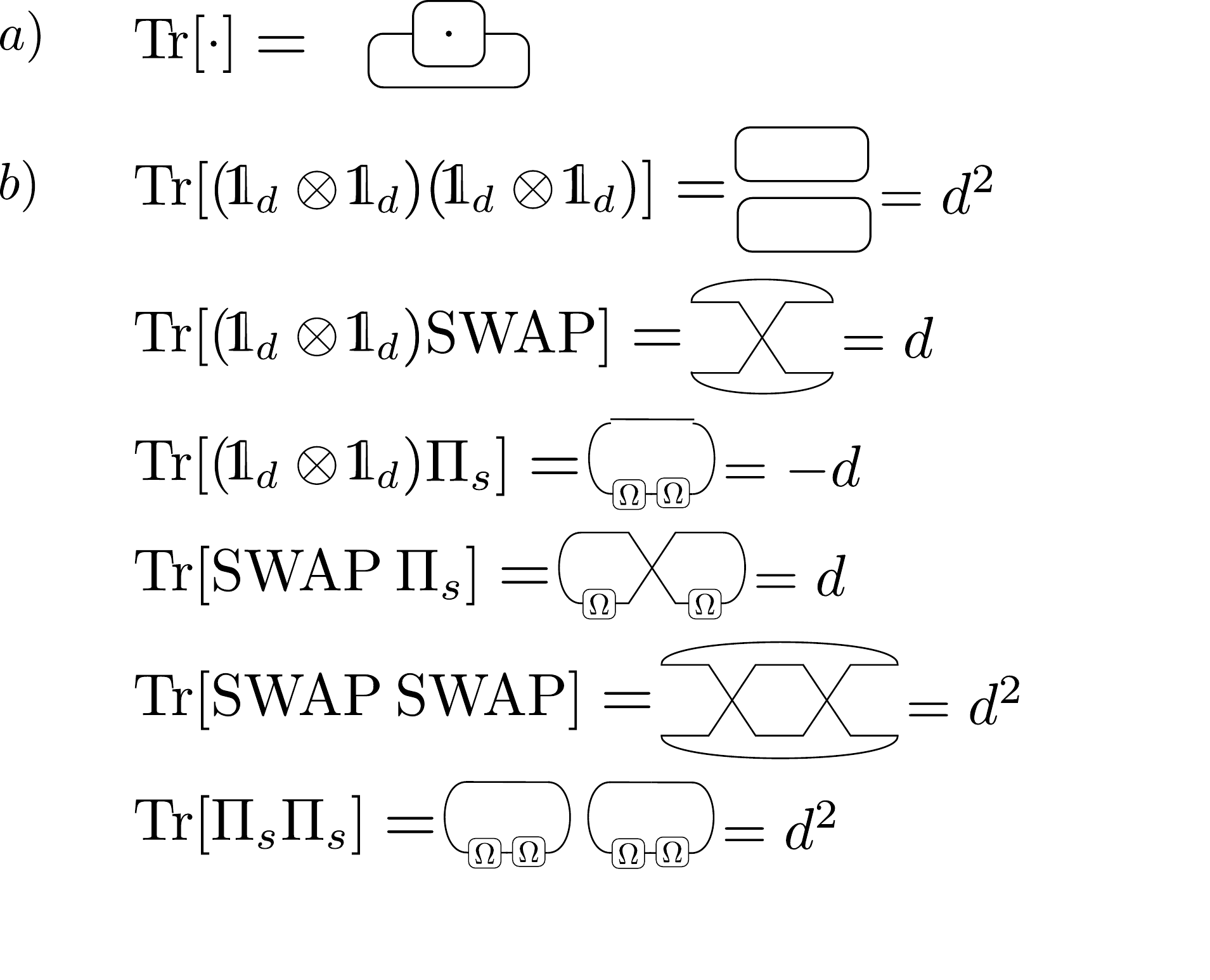}
    \caption{{\bf Computation of the Gram matrix elements}. a) Tensor representation of the trace operation. b)  We diagrammatically  show how to compute all the matrix entries of W for $t=2$, using tensor notation. The computation of these matrix elements leads to Eq.~\eqref{eq:gram-t2}.}
    \label{fig:gram}
\end{figure}

Equation~\eqref{eq:gram-split} allows us to write 
\begin{equation} \label{eq:W-matrix}
     W^{-1}=\frac{1}{d^t}\left(\id_{D_t}+C\right)\,,
\end{equation}
where the matrix entries of $C$ are in $\OC(1/d)$. To see this, suppose we write $B$ in a basis such that it is diagonal. Then, in this basis $W^{-1}$ is diagonal with entries $[W^{-1}]_{\mu\mu}=\frac{1}{d^t(1+\lambda_\mu/d)}=\frac{1}{d^t}(1-\frac{\lambda_\mu}{d+\lambda_\mu})$, so let us write it as $W^{-1}=\frac{1}{d^t}\left(\id_{D_t}+D\right)$. Since all the matrix entries of $B$ are in $\OC(1)$, so are its eigenvalues, i.e., $\lambda_\mu\in\OC(1)$ $\forall \mu$, which implies that the entries of $D$ are at most $\OC(1/d)$. Finally, $C$ and $D$ are related by a unitary change of basis, and therefore the matrix entries of $C$ are suppressed as $\OC(1/d)$. 
Once we have found  $W^{-1}$, all that is left is to evaluate Eq.~\eqref{eq:twirl-expansion}, which leads to
\small
\begin{equation} \label{eq:twirl-t}
\begin{split}
    \TC^{(t)}_{\mathbb{SP}(d/2)}\left[X\right]&=\frac{1}{d^t}\sum_{\sigma\in \brauer}\Tr\left[X F_d(\sigma^T)\right]\,F_d(\sigma)\\&+\frac{1}{d^t}\sum_{\sigma,\pi\in \brauer}c_{\pi,\sigma}\,\Tr\left[XF_d(\sigma^T)\right]\,F_d(\pi)\,,
\end{split}
\end{equation}
\normalsize
where the $c_{\pi,\sigma}$ are the matrix entries of $C$ in Eq.~\eqref{eq:W-matrix}, and thus are upper bounded as $\OC(1/d)$. Moreover, we here recall that given some $\sigma\in \brauer$, we define its transpose as $\sigma^T=\{\{\lambda_{1}+t, \sigma(\lambda_{1})+t\}\cup\dots\cup\{\lambda_{t}+t , \sigma(\lambda_{t})+t\}\}$, where the sum is taken mod $t$.
Note that if we fix $t$ and take the limit $d\rightarrow\infty$, the second sum in Eq.~\eqref{eq:twirl-t} gets asymptotically suppressed with the Hilbert space dimension $d$.

\section{Gaussian processes from random symplectic circuits}

Recently it has been shown that Pauli measurement outcomes at the output of Haar random circuits sampled from  $\mathbb{U}(d)$ or $\mathbb{O}(d)$~\cite{garcia2023deep,rad2023deep} (as well as the outputs of some shallow quantum neural networks~\cite{girardi2024trained}) converge in distribution to Gaussian Processes (GPs) under certain assumptions. In this section we will show that the asymptotic Weingarten tools previously presented can be used to prove that such phenomenon will also occur for Haar random symplectic quantum circuits.

\begin{figure}[t]
    \centering
    \includegraphics[trim={3.5cm 0.5cm 2.5cm 2cm},clip, width=0.45\textwidth]{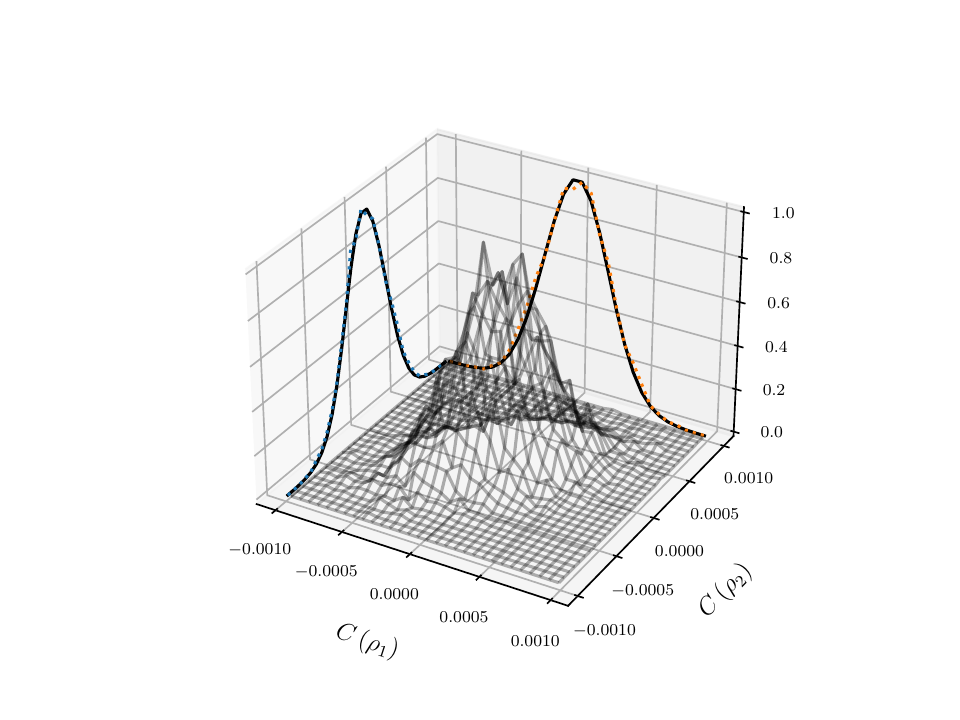}
    \caption{{\bf Random symplectic circuits form Gaussian processes}. We show the emergence of a Gaussian process for the set of pure states $\rho_1=\ketbra{\psi_1}{\psi_1}$ and $\rho_2=\ketbra{\psi_2}{\psi_2}$, with $\ket{\psi_1}=\ket{0}^{\otimes n}$ and $ \ket{\psi_2}=\frac{1}{\sqrt{2}}\left(\ket{0}^{\otimes n} + \ket{0}\ket{1}\ket{0}^{\otimes n-2}\right)$, where the number of qubits is $n=24$. The measurement operator is $O=Y_2$.     The statistics are obtained by sampling $10^4$ independent unitaries from $\mathbb{SP}(d/2)$. Together with the three-dimensional normalized histogram of sampled values of $C(\rho_j)$, we plot the separate samples counts for the two states considered (orange and blue scatter data) overlaid with the predicted Gaussian distribution.}
    \label{fig:symplectic-gp}
\end{figure}

In particular,  we will consider a setting where we are given a set  $\mathscr{D}=\{\rho_1,\dots,\rho_m\}$ of $n$-qubit quantum states on a $d$-dimensional Hilbert space. We then take the  $m$ states from $\mathscr{D}$ and  send them through a unitary $U$ which is sampled according to the Haar measure over $\mathbb{SP}(d/2)$. At the output of the circuit we measure the expectation value of a Pauli operator $O$ taken from $i\mathfrak{sp}(d/2)$~\footnote{Our results also hold if instead of a Pauli $O$ from $i\mathfrak{sp}(d/2)$ we take $O'=S O S^{\dagger}$, with $S$ an arbitrary matrix from $\SPBB(d/2)$.}. This leads to a set of quantities of the form 
 \begin{equation}\label{eq:ci}
    C(\rho_j)=\Tr[U\rho_j U\ad O]\,,
\end{equation}
which we collect in a length-$m$ vector
\begin{equation}\label{eq:random_vector} 
\mathscr{C}=\left(C(\rho_1),\ldots, C(\rho_m) \right) \,. 
\end{equation}

We will say that $\mathscr{C}$ forms a GP iff it follows a multivariate Gaussian  distribution, which we denote as $\NC(\vec{\mu},\vec{\Sigma})$.\footnote{Alternatively, we can also say that $\mathscr{C}$ forms a GP iff every linear combination of its entries  follows a univariate Gaussian distribution.}  We recall that a  multivariate Gaussian $\NC(\vec{\mu},\vec{\Sigma})$ is completely determined by its $m$-dimensional mean vector $\vec{\mu}=(\mathbb{E}[C(\rho_1)],\ldots,\mathbb{E}[C(\rho_m)])$, and its $m\times m$ dimensional covariance matrix with entries $\vec{\Sigma}_{jj'}={\rm Cov}[C(\rho_{j}),C(\rho_{j'})]$, as all  higher moments can be computed from $\vec{\mu}$ and $\vec{\Sigma}$ alone via Wick's theorem~\cite{isserlis1918formula}. Hence, in what follows we will determine conditions for which  $\mathscr{C}$ forms a GP, and report only its mean and its covariance matrix entries.

First, we can show that the following theorem holds (see Appendix~\ref{ap:th-2} for a proof).
\begin{theorem}\label{th:gauss-1}
    Let $\mathscr{C}$ be a vector of expectation values of the Hermitian operator $O$ over a set of states from $\mathscr{D}$, as in Eq.~\eqref{eq:random_vector}. If  $\Tr[\rho_j \rho_{j'}]\in\Omega\left(1/\poly(\log(d))\right)$ and $\Big|\Tr[\Omega\rho_j \Omega\rho_{j'}^T]\Big|\in o\left(1/\poly(\log(d))\right)$ $\forall j,j'$, then in the large $d$-limit $\mathscr{C}$ forms a GP with mean vector $\vec{\mu}=\vec{0}$ and covariance matrix
\begin{equation}\label{eq:covariance-gp1}
    \vec{\Sigma}_{j,j'} = \frac{\Tr[\rho_{j}\rho_{j'}]}{d} \,.
\end{equation}
\end{theorem}

Interestingly, we see here that the states for which Theorem~\ref{th:gauss-1} holds and $\mathscr{C}$ forms a GP are such that their inner products are at most polynomially vanishing with $n$. However, when we conjugate them by $\Omega = iY\otimes\id_{d/2}$, effectively leading to rotated states  $ \widetilde{\rho_j}= Y_1 \rho_j Y_1$ (up to a minus sign), then the inner products between the $\widetilde{\rho_j}$ and the original $\rho_{j'}^T$ are strictly smaller than  polynomially vanishing with $n$. We use precisely this condition to create a set $\mathscr{D}$ of such states in Fig.~\ref{fig:symplectic-gp}, where we show that the distribution indeed converges to a multivariate Gaussian with positive correlation. In particular, we there consider a system of $n=24$ qubits and sample $10^4$ independent unitaries from $\mathbb{SP}(d/2)$.\footnote{Sampling from $\mathbb{SP}(d/2)$ can be achieved by initializing a random $\frac{d}{2}\times\frac{d}{2}$ quaternionic matrix, mapping it to its complex $d\times d$ representation and performing a QR decomposition of the latter, as explained in \cite{mezzadri2006generate}.}

Next, we are also able to prove convergence to a GP in a different regime. Namely, when the overlaps between $\rho_j$ (as well as between its transformed version $\widetilde{\rho_j}$) and  $\rho_{j'}$ remain at most polynomially vanishing. This result is  stated in the next theorem, whose proof we present in Appendix~\ref{ap:th-3}.

\begin{theorem}\label{th:gauss-2}
    Let $\mathscr{C}$ be a vector of expectation values of the Hermitian operator $O$ over a set of states from $\mathscr{D}$, as in Eq.~\eqref{eq:random_vector}. If  $\Big|\Tr[\rho_j \rho_{j'}]+\Tr[\Omega\rho_j \Omega\rho_{j'}^T]\Big|\in\Omega\left(1/\poly(\log(d))\right)$  $\forall j,j'$, then in the large $d$-limit $\mathscr{C}$ forms a GP with mean vector $\vec{\mu}=\vec{0}$ and covariance matrix
\begin{equation}\label{eq:covariance-gp2}
    \vec{\Sigma}_{j,j'} = \frac{2\,\Tr_{\mathfrak{g}}[\rho_{j}\rho_{j'}]}{d} \,.
\end{equation}
\end{theorem}
Here, we defined $\Tr_{\mathfrak{g}}[\rho_j\rho_{j'}]$ as follows. Given that any quantum state can be written as $\rho=\frac{1}{d} \sum_k c_k P_k$, where the sum runs over all $d^2$ Pauli matrices (including the identity) and $-1\leq c_k\leq 1$ $\forall k$, it follows that 
\begin{equation}
    \rho= \frac{1}{d} \left(\sum_{k \,\backslash\,  P_k\in i\mathfrak{sp}(d/2)} c_k P_k +\sum_{k \,\backslash\,  P_k\notin i\mathfrak{sp}(d/2)} c_k P_k \right) \,, 
\end{equation}
where we separated $\rho$ into its algebra and out-of-the-algebra components, i.e., $\rho=\rho_{\mathfrak{g}} + \rho_{\overline{\mathfrak{g}}}$. Then, $\Tr_{\mathfrak{g}}[\rho_j\rho_{j'}]$ is simply the Hilbert-Schmidt product between the algebra components of $\rho_j$ and $\rho_{j'}$. For instance, 
\begin{equation} \label{eq:g-covariance}
    \Tr_{\mathfrak{g}}[\rho_j\rho_{j}]= \frac{1}{d}\sum_{k \,\backslash\,  P_k\in i\mathfrak{sp}(d/2)} c_k^2\,.
\end{equation}

Note that a crucial  difference between Eqs.~\eqref{eq:covariance-gp1} and~\eqref{eq:covariance-gp2} is that in the former all covariances must be positive (i.e., we have a positively correlated GP), whereas in the latter the covariances can be negative. Finally, we prove in Appendix~\ref{ap:th-4} that there exist states for which symplectic quantum circuits  form uncorrelated GPs.

\begin{theorem}\label{th:gauss-3}
    Let $\mathscr{C}$ be a vector of expectation values of the Hermitian operator $O$ over a set of states from $\mathscr{D}$, as in Eq.~\eqref{eq:random_vector}. If  $\Tr[\rho_j^2]+\Tr[\Omega\rho_j \Omega\rho_j^T]\in\Omega\left(1/\poly(\log(d))\right)$  $\forall j$ and $\Tr[\rho_j \rho_{j'}]=-\Tr[\Omega\rho_j \Omega\rho_{j'}^T]$ $\forall j\neq j'$, then in the large $d$-limit $\mathscr{C}$ forms a GP with mean vector $\vec{\mu}=\vec{0}$ and diagonal covariance matrix
\begin{equation}\label{eq:covariance-gp3}
    \vec{\Sigma}_{j,j'} = \begin{cases}
        \frac{2\,\Tr_{\mathfrak{g}}[\rho_{j}^2]}{d} \quad &{\rm if\; } j=j' \\
         \; 0 \qquad \quad\;\,\, &{\rm if\; } j\neq j'
    \end{cases} \,.
\end{equation}
\end{theorem}

Let us remark that to understand the technical conditions in Theorems~\ref{th:gauss-2} and~\ref{th:gauss-3}, we can use the following equation (proven in Appendix~\ref{ap:th-2}),
\begin{equation} 
    \Tr[\rho_j \rho_{j'}]+\Tr[\Omega\rho_j \Omega \rho_{j'}^T ] = 2 \,\Tr_{\mathfrak{g}}[\rho_j \rho_{j'}] \,.
\end{equation} 
It follows that in order for Theorem~\ref{th:gauss-2} to hold, we require that the (absolute value of the) inner products between the algebra components of $\rho_j$ and $\rho_j'$ are at most polynomially vanishing in $\log(d)$ for all $j,j'$.
Analogously, in Theorem~\ref{th:gauss-3}, the condition is that the algebra components of $\rho_j$ are at most polynomially vanishing in $\log(d)\;\forall j$, while the inner product of the algebra components of $\rho_j$ and $\rho_j'$ is exactly zero $\forall j\neq j'$.

\section{Concentration of measure in symplectic circuits}

In this section, we show that we can leverage the knowledge of the exact output distribution of random symplectic quantum circuits to characterize the concentration of measure phenomenon in these circuits. In particular, we provide concentration bounds for circuits that are sampled from the Haar measure on the symplectic group, and also for circuits that form $t$-designs over $\SPBB(d/2)$. In the case of Haar random symplectic circuits we compute tail probabilities to obtain the desired bound. For symplectic $t$-designs, we use an extension of Chebyshev's inequality to arbitrary moments. 
Our first result is:

\begin{corollary}\label{cor:double}
    Let $C(\rho_j)$ be the expectation value of a Haar random symplectic quantum circuit as in Eq.~\eqref{eq:ci}. If the conditions under which Theorems~\ref{th:gauss-1}, ~\ref{th:gauss-2} and~\ref{th:gauss-3} hold are satisfied, then 
    \begin{equation}
      {\rm Pr}(|C(\rho_j)|\geq c)\in\OC\left(\frac{\Tr_\g\left[\rho_j^2\right]}{c\sqrt{d} e^{dc^2 / (4\,\Tr_\g[\rho_j^2])}}\right)\,.
    \end{equation}
\end{corollary}

This corollary, proven in Appendix~\ref{ap:cor-1}, shows that Haar random symplectic processes concentrate exponentially in the Hilbert space dimension. That is, doubly-exponentially in the number of qubits. This feature is analogous to that encountered in random unitary and orthogonal quantum circuits~\cite{popescu2006entanglement,garcia2023deep}. Intuitively, we can understand this result from the fact that the probability density function of a Gaussian distribution decreases exponentially with $1/ \sigma^2$, and here $\sigma^2$ is itself exponentially decreasing with the number of qubits (see Eq.~\eqref{eq:covariance-gp1}). We also remark that the smaller the component of $\rho_j$ in the algebra, i.e., the smaller $\Tr_\g[\rho_j^2]$ is, the more concentrated $C(\rho_j)$ becomes, in agreement with the results in Refs.~\cite{fontana2023theadjoint,ragone2023unified}.

This extreme concentration of measure comes at the expense of the exponential (in the number of qubits) time or depth that is required to obtain a truly Haar random circuit~\cite{knill1995approximation}. In practice, however, one often encounters circuits that are not fully random but instead form $t$-designs. Therefore, if $U$ forms a $t$-design over $\mathbb{SP}(d/2)$, we can provide tight concentration bounds for the circuit outputs (see Appendix~\ref{ap:cor-2} for a proof).

\begin{corollary} \label{cor:t-designs}
Let $C(\rho_j)$ be the expectation value of a quantum circuit that forms a $t$-design over $\SPBB(d/2)$ as in Eq.~\eqref{eq:ci}. If the conditions under which Theorems~\ref{th:gauss-1}, ~\ref{th:gauss-2} and~\ref{th:gauss-3} hold are satisfied, then 
\small
\begin{equation} \label{eq:concentration}
    {\rm Pr}(|C(\rho_j)|\geq c)\in\OC\left(\left(2\left\lfloor t/2\right\rfloor-1\right)!!\, \left(\frac{2\,\Tr_\g\left[\rho_j^2\right]}{d c^2}\right)^{\left\lfloor t/2\right\rfloor }\right)\,.
\end{equation}
\normalsize
\end{corollary}

This result implies that symplectic $2$-designs will concentrate as $\OC(1/d)$, $4$-designs as $\OC(1/d^2)$, $6$-designs as $\OC(1/d^3)$, etc.
Note that for $t=2$ the bound in Eq.~\eqref{eq:concentration} is analogous to the known concentration result for unitary $2$-designs (i.e., barren plateaus~\cite{mcclean2018barren,cerezo2020cost,larocca2024review}).

\section{Anti-concentration in symplectic circuits}

Let us now  study the emergence of anti-concentration in symplectic quantum circuits that form $2$-designs over the $\SPBB(d/2)$ group. Anti-concentration roughly refers to the property that the output probabilities after measuring in the computational basis are not concentrated in a small subset of bit-strings~\cite{dalzell2022randomquantum,hangleiter2018anticoncentration}. More precisely, we say that quantum circuits sampled from a measure $\mu$ (e.g., the Haar measure) on a set of unitaries exhibit anti-concentration when there exist constants $\alpha,\beta>0$ such that
\begin{equation} \label{eq:anti-concentration}
    {\rm Pr}_{U\sim \mu}\left(\left|\bra{x} U \ket{0}^{\otimes n}\right|^2 \geq \frac{\alpha}{d}\right) > \beta\,,
\end{equation}
for all computational-basis states $\ket{x}$, with $x\in\{0,1\}^n$. That is, the probability that for a unitary $U$ sampled from $\mu$, all bit-string probabilities are at most a non-zero constant factor away from the uniform distribution, is lower-bounded by a positive constant.
Anti-concentration has been shown to be a very important property, as it is a necessary condition for the hardness of classical simulation in random circuit sampling~\cite{chen2021exponential,boixo2018characterizing,dalzell2021random,oszmaniec2022fermion,arute2019quantum,bouland2019complexity,kondo2022quantum,movassagh2023hardness}. Here, we prove that Haar random symplectic circuits and circuit ensembles that form symplectic $2$-designs anti-concentrate, as stated in the following theorem.

\begin{theorem} \label{th:anti-con} Let $\mu$ be the Haar measure on $\SPBB(d/2)$, or a measure giving rise to a $2$-design over $\SPBB(d/2)$. Then,
    \begin{equation}
        {\rm Pr}_{U\sim \mu}\left(\left|\bra{x} U \ket{0}^{\otimes n}\right|^2 \geq \frac{\alpha}{d}\right) \geq \frac{(1-\alpha)^2}{2}  \,,
    \end{equation}
for $0\leq\alpha\leq 1$.
\end{theorem}

A detailed proof of this theorem can be found in Appendix~\ref{ap:th-5}. We note that the anti-concentration result can also be understood from the so-called collision probability~\cite{dalzell2022randomquantum}, defined as
\begin{equation}\label{eq:colision}
    z = \mathbb{E}_{U\sim \mu}\left[\sum_{x \in \{0,1\}^{\otimes n}} p_U(x)^2\right]=2^n\mathbb{E}_{U\sim \mu}\left[p_U(x)^2\right]\,,
\end{equation}
where $p_U(x) = \left|\bra{x} U \ket{0}^{\otimes n}\right|^2$, which can be found to be equal to 
\begin{equation}
    z_H = \frac{2}{d + 1}\,,
\end{equation}
when $\mu$ is the Haar measure over $\SPBB(d/2)$. Hence, we can see that the probability measurements are indeed at most a non-zero constant factor away from the uniform distribution (for which $z=1/d$).

\section{Shallow locally random  symplectic circuits}\label{sec:symp-anticoncentration}

\begin{figure}[t]
    \centering
    \includegraphics[width=.25\textwidth]{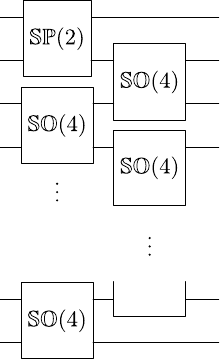}
    \caption{{\bf Shallow locally random  symplectic circuits}. Architecture employed to study the properties of shallow random symplectic circuits. Here, the notation $\SPBB(2)$ ($\mathbb{SO}(4)$) implies that the unitary is sampled from the Haar measure over $\SPBB(2)$ ($\mathbb{SO}(4)$). The gates shown in the picture, arranged in a brick-layered fashion and acting on neighboring qubits in a one-dimensional lattice, define one layer for the numerical experiments presented in Sec.~\ref{sec:symp-anticoncentration}. }
    \label{fig:local-random-circ}
\end{figure}

In the previous sections we have discussed tools to work with random quantum circuits that are Haar random, or that form a $t$-design over $\mathbb{SP}(d/2)$. However, one might be interested in studying the properties of shallow random circuits sampled, according to some measure $dU$, from some set of unitaries  $\SC\subseteq \mathbb{SP}(d/2)$. Here, one needs to evaluate $t$-th order twirls such as
\begin{equation} \label{eq:twirl-set}
    \TC^{(t)}_{\SC}\left[ X\right] =\int_{\SC} dU\, U^{\otimes t} X (U^\dagger)^{\otimes t} \,,
\end{equation}
or concomitantly, we need to compute the $t$-th moment operator
\begin{equation} \label{eq:moment-set}
    \widehat{\TC}^{(t)}_{\SC} =\int_{\SC} dU\, U^{\otimes t} \otimes (U^*)^{\otimes t} \,.
\end{equation}
Now, given that $\SC$ need not be a group, one cannot directly leverage the Weingarten calculus to evaluate these quantities. However, the analysis of $ \TC^{(t)}_{\SC}$ and $\widehat{\TC}^{(t)}_{\SC}$ can  become tractable again under the assumption that the circuit is composed of  gates that are  sampled according to the Haar measure from some local group. In particular,  let us consider a circuit which takes the form
\begin{equation}
    U = \prod_l U_l\,,
\end{equation}
where we assume that $U_l$ acts non-trivially on $k_l$ (non-necessarily neighboring) qubits whose indexes we denote as $I_l$, and where we omitted the parameter dependency for ease of notation. For instance, if $U_1$ is a three-qubit gate acting on the first, second, and third qubits, we would have $k_1=3$ and $I_1=\{1,2,3\}$. Then, it is standard to assume that each $U_l$ is independently sampled from a local group $G_l\subseteq\mathbb{U}(2^{k_l})$ according to its associated Haar measure $d\mu_l(U_l)$. In this scenario, we can see that if we write
\begin{equation}\label{eq:moment-product}
\widehat{\TC}^{(t)}_{\SC}=\prod_j\widehat{\TC}^{(t)}_{G_l}\,,  
\end{equation}
where 
\begin{equation}
    \widehat{\TC}^{(t)}_{G_l}=\int_{G_l} d\mu_l(U_l)\, U_l^{\otimes t} \otimes (U_l^*)^{\otimes t}\,,
\end{equation}
then each individual $t$-th moment operator $\widehat{\TC}^{(t)}_{G_l}$ associated to each local gate can be evaluated via the Weingarten calculus as a projector onto the $t$-th fold commutant of $G_l$. 

\begin{figure}[t]
    \centering
    \includegraphics[width=0.45\textwidth]{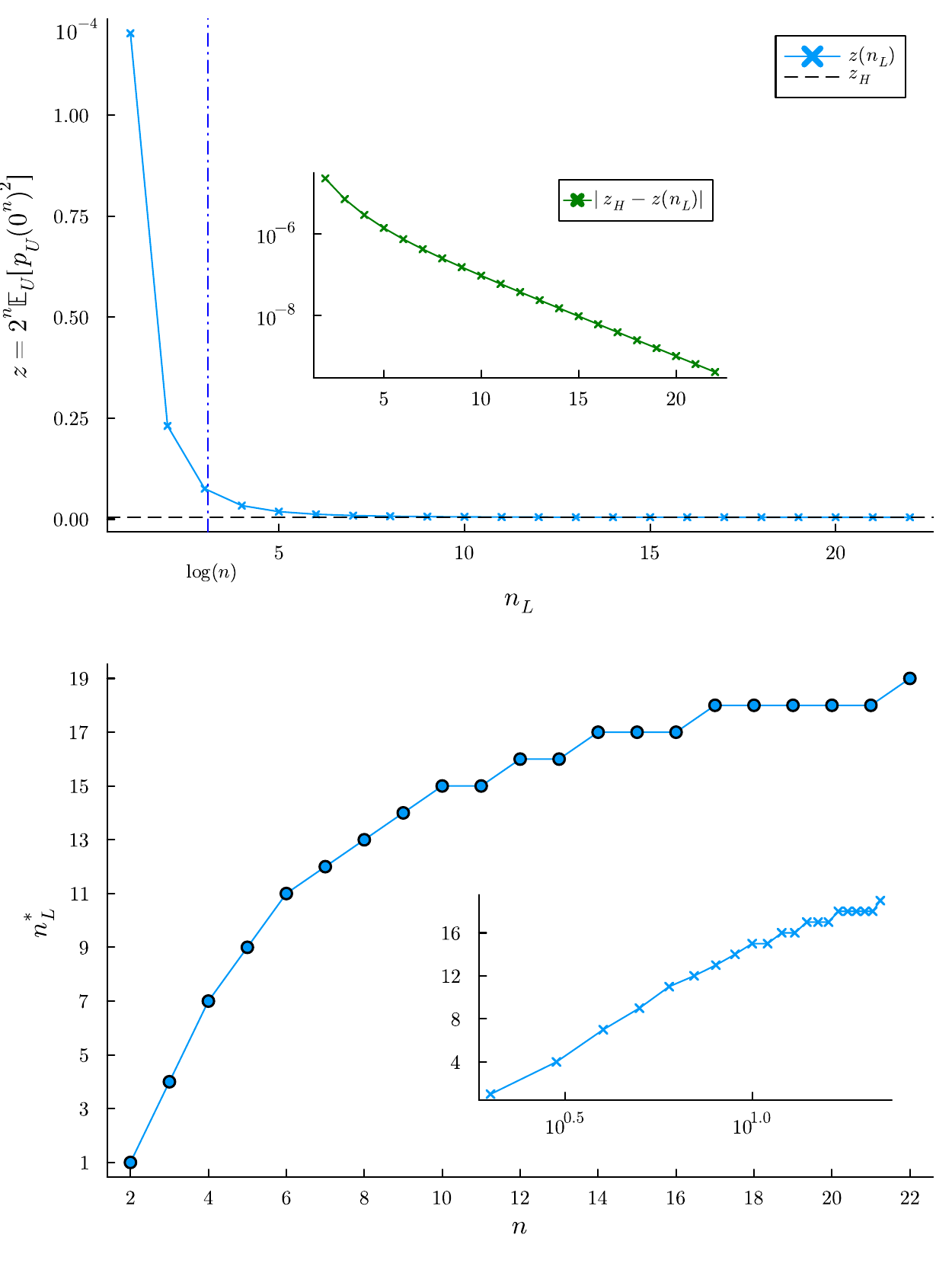}
    \caption{{\bf Random symplectic quantum circuits anti-concentrate}. The top panel shows the scaling of the collision probability $z$ (blue line) as a function of the number of layers $n_L$, for a system size of $n=22$ qubits. The dashed black line is the value $z_H$ to which $z$ converges as the circuit converges to a 2-design. The inset shows the difference $\vert z_H-z\vert$. In the bottom panel we show the number of layers $n_L^*$ needed to attain $\vert z_H-z\vert<\frac{\epsilon}{d}$, where we chose $\epsilon=0.01$, for different system sizes ranging from 2 to 22 qubits. The inset shows the same data in log-x scale, providing evidence for symplectic circuits anti-concentration happening at logarithmic depth.}
    \label{fig:anti-concentration}
\end{figure}

The previous idea of studying  circuits composed of random local gates has been explored in Refs.~\cite{hayden2007black,sekino2008fast,brown2012scrambling,lashkari2013towards,hosur2016chaos,nahum2018operator,von2018operator,hunter2019unitary,barak2020spoofing,napp2022quantifying,harrow2018approximate,letcher2023tight,hayden2016holographic,nahum2018operator,hunter2019unitary,harrow2018approximate,cerezo2020cost,pesah2020absence,braccia2024computing}, and it has been shown that the task of computing the moment operator can be mapped to that of analyzing  a Markov chain-like process obtained from the product of the non-orthogonal projectors $\widehat{\TC}^{(t)}_{G_l}$. Importantly, it is worth highlighting the fact that most of the previous references work with random quantum circuits composed of gates sampled from $\mathbb{U}(4)$ (with the notable exception of~\cite{braccia2024computing}). Thus, little to-no-attention has been payed to local random circuits leading to sets of unitaries  $\SC$ that belong to the symplectic group. As such, the question of who the local groups $G_l$ can be so that one still obtains (globally) symplectic unitaries in $\SC$ has not been yet addressed.

While \textit{a priori} one could be tempted to choose all $G_l$ as $\mathbb{SP}(2^{k_l-1})$, this would lead to non-symplectic unitaries (as we have already seen that arbitrary unitaries can be constructed from locally symplectic gates, see Sec.~\ref{sec:local-symp-circs}). Instead, referring to Proposition~\ref{prop:sp-algebra} we find that the most natural choice is
\begin{equation}\label{eq:groups-local}
    G_l=\begin{cases}
    \mathbb{O}(2^{k_l})\,,\quad\quad\,\, \text{if } 1\notin I_l\,,\\
    \mathbb{SP}(2^{k_l-1})\,,\quad \text{if } 1\in I_l\,.\\
    \end{cases}
\end{equation}
That is, if the gate $U_l$ acts non-trivially on the first qubit, then it must be sampled from a symplectic local group $\mathbb{SP}(2^{k_l-1})$, while if it does not act on the first qubit, then it must be sampled from an orthogonal local group $\mathbb{O}(2^{k_l})$. In fact, it follows directly from the proof of Theorem~\ref{th:symplectic-universal} that such circuits will produce unitaries in $\mathbb{SP}(d/2)$.

Given Eq.~\eqref{eq:groups-local}, one can  analyze features of locally random circuits such as how fast will their properties converge to those of a $t$-design over $\mathbb{SP}(d/2)$. As an example, let us study the depth at which the probability outcomes in the computational basis anti-concentrate, for a circuit composed of two-qubit Haar random gates acting in a brick-layered fashion on neighboring qubits (see Fig.~\ref{fig:local-random-circ}). Here, we need to evaluate the second moment operators $\widehat{\TC}^{(2)}_{\mathbb{SP}(2)}$ and $\widehat{\TC}^{(2)}_{\mathbb{O}(4)}$,  which respectively project onto their commutants spanned by $\{\id_4\otimes \id_4,SWAP,\Pi_s\}$ and  $\{\id_4\otimes \id_4,SWAP,\Pi\}$, with $\Pi=(\id_d\otimes\Omega)\Pi_s(\id_d\otimes\Omega)$. From here, we can study the action of $\widehat{\TC}^{(2)}_{\mathbb{SP}(2)}$ and $\widehat{\TC}^{(2)}_{\mathbb{O}(4)}$ simply by studying how they project their local commutants onto each other. In particular, we can leverage the recently developed tensor-network formalism of Ref.~\cite{braccia2024computing} to numerically investigate the behavior of the collision probability $z$ defined in Eq.~\eqref{eq:colision}.

We will analyze how $z$ changes as a function of the number of layers $n_L$ in the circuit (where a layer is defined as in Fig.~\ref{fig:local-random-circ}) for different qubit numbers, as this scaling can be used to diagnose the depth at which the architecture at hand anti-concentrates. In Fig.~\ref{fig:anti-concentration}, we show how $z$ approaches $z_H$ with increasing circuit depth for a system of $n=22$ qubits. There we also present the depth $n_L^*$ for which the  difference $\vert z_H-z\vert$ becomes  smaller than $\epsilon/d$, for some small constant $\epsilon$. In particular, we deem the condition $\vert z_H-z\vert<\frac{\epsilon}{d}$ as the emergence of anti-concentration \cite{dalzell2022randomquantum}. Our numerical results show that anti-concentration happens at logarithmic depth, i.e., $n_L^*=\OC(\log (n))$, which is the same scaling observed for quantum circuits composed of random unitary and orthogonal local gates \cite{dalzell2022randomquantum,braccia2024computing}.

\section{Conclusions}

In this work, we have addressed the study of quantum circuits implementing symplectic unitary transformations. In particular, we have introduced a simple universal architecture for symplectic unitaries that can be readily implemented on near-term quantum hardware, as it only requires one- and two-qubit gates acting on nearest neighbors in a one-dimensional lattice. Furthermore, we have derived properties of random symplectic circuits, both in the deep and shallow regimes, including a proof that the circuits' outputs can converge  to Gaussian processes (e.g., when measuring a Pauli) or exhibit the anti-concentration phenomenon (when performing computational-basis measurements). It remains to be proven whether these circuits anti-concentrate at logarithmic depth, although our numerical experiments suggest that this is indeed the case. In addition, we stress that the fact that symplectic quantum circuits form Gaussian processes opens the door to applying Bayesian inference techniques for learning tasks~\cite{garcia2023deep}.

Interestingly, our work reveals some key differences between circuits that implement unitary or orthogonal evolutions, and those that implement symplectic ones. For instance, we have shown that the  structure of the symplectic Lie algebra and its associated Lie group places a privileged role on a single qubit in the system, thus breaking typical qubit-exchange symmetries appearing when working with $\mathbb{U}(2^n)$ or $\mathbb{O}(2^n)$. This small, albeit important difference makes it such that care must be taken when constructing symplectic circuits, as translationally invariant sets of generators are not available. It also leads to potentially counter-intuitive results, such as circuits composed of locally symplectic gates being able to produce non-symplectic unitaries. 

Looking forward, we expect that our constructions will encourage the community to explore the simulation of physical processes described via symplectic unitaries, as these can now be compiled to qubit architectures and therefore implemented in most currently-available quantum hardware. Indeed, we hope that our work will spark the interest on quantum circuits that produce symplectic evolutions, and that compelling applications will be discovered soon.

\section*{Acknowledgments}

We thank Martin Larocca, Bojko N. Bakalov, Nahuel L. Diaz, and Alexander F. Kemper for insightful  conversations. D.G.M., P.B. and M.C.  were supported by Laboratory Directed Research and Development (LDRD) program of Los Alamos National Laboratory (LANL) under project numbers 20230527ECR and 20230049DR. M.C. was also initially  supported by LANL's ASC Beyond Moore’s Law project.

\section*{}

{\bf Note added}: After the completion of this work,  it was found that symplectic quantum circuits can be employed to efficiently solve BQP-complete problems directly related to the simulation of Gaussian bosonic circuits on exponentially many modes, thus obtaining an exponential quantum advantage~\cite{barthe2024gate}. Furthermore, it has been discovered that the orbits of random unitary circuits and random symplectic circuits are \emph{unconditionally} indistinguishable for any $t$, and any number of measurement shots. That is, the ensemble of \emph{pure} quantum states that one obtains by applying $U^{\otimes t}$ to an arbitrary pure state $|\psi\rangle^{\otimes t}$ is completely identical irrespective of whether $U$ is drawn from the unitary or symplectic Haar measures~\cite{west2024random}.

\bibliography{quantum.bib}

\appendix
\setcounter{theorem}{0}
\setcounter{corollary}{0}
\setcounter{proposition}{0}

\section{Proof of Proposition 1}
\label{ap:prop-1}

In this appendix, we provide the proof of Proposition~\ref{prop:sp-algebra}, which we recall for convenience.

 \begin{proposition} 
     A basis for the standard representation of the $\mathfrak{sp}(d/2)$ algebra is 
     \begin{equation} \label{eq-ap:sp-dla}
	   B_{\mathfrak{sp}(d/2)}\equiv i\{\{X,Y,Z\}\otimes P_s \}\,\cup\, i\{\id\otimes P_a\}\,,
    \end{equation}
    where $P_s$ and $P_a$ belong to the sets of arbitrary symmetric and anti-symmetric Pauli strings on $n-1$ qubits, respectively, and $\id,X,Y,Z$ are the usual $2\times 2$ Pauli matrices.
 \end{proposition}
 
\begin{proof}
We first recall that any matrix $M\in\mathfrak{sp}(d/2)$ satisfies $M^T\Omega=-\Omega M$. Then, we note that $\Omega= iY\otimes\id_{d/2}$. Hence, we are looking for Pauli strings $P$ satisfying  $P^T (Y\otimes\id_{d/2})=-(Y\otimes\id_{d/2}) P$. We know that $P^T=P$ if and only if the number of $Y$'s in $P$ is even, and $P^T=-P$ otherwise. Therefore, the Pauli matrices belonging to $i\mathfrak{sp}(d/2)$ have to anti-commute with $(Y\otimes\id_{d/2})$ when $P^T=P$, and commute with it when $P^T=-P$. Given that any Pauli commutes with $\id_{d/2}$, it follows that all Pauli strings in $i\mathfrak{sp}(d/2)$ have the form $\{X,Y,Z\}\otimes P_s$ or $\id\otimes P_a$,
with $P_s$ a symmetric and $P_a$ an anti-symmetric Pauli string. Since $M^T\Omega=-\Omega M$ is a linear equation, all (real) linear combinations of such Paulis also belong in $i\mathfrak{sp}(d/2)$. Finally, the dimension of $B_{\mathfrak{sp}(d/2)}$ can be checked to be $d(d+1)/2$, which is precisely $\dim(\mathfrak{sp}(d/2))$, as follows. The set
$\{P_s\}$ $\left(\{P_a\}\right)$ is composed of all Paulis acting on $n-1$ qubits with an even (odd) number of $Y$'s, and it is an orthogonal basis for the space of symmetric (anti-symmetric) matrices. Therefore, the dimensions of $\{P_s\}$ and $\{P_a\}$ are $N_1=2^{n-1}(2^{n-1}+1)/2$ and $N_2=2^{n-1}(2^{n-1}-1)/2$, respectively.
 As such, there are $3N_1+N_2=2^{n-1}(2^n+1)=d(d+1)/2$ elements in $B_{\mathfrak{sp}(d/2)}$, which is precisely the dimension of $\mathfrak{sp}(d/2)$.
We then conclude that $B_{\mathfrak{sp}(d/2)}$ is a basis of $\mathfrak{sp}(d/2)$.

\end{proof}

 As a sanity check, we can show that the operators in Eq.~\eqref{eq-ap:sp-dla} are indeed closed under commutation, so that they form a Lie algebra. For this purpose, we can make use of the properties of the commutator of symmetric and anti-symmetric matrices. In particular, let $P_a, P_a',P_a''$ be anti-symmetric Pauli strings and $P_s, P'_s,P_s''$   symmetric Pauli strings. Then, we have that $([P_a,P_a'])^T = (P_a P_a' - P_a'P_a)^T= P_a' P_a - P_a P_a'= -[P_a,P_a']$. Similarly, we find that $([P_s,P_s'])^T= -[P_s,P_s']$ and $([P_a,P_s])^T=[P_a,P_s]$. That is, the commutator of two anti-symmetric or symmetric Paulis is anti-symmetric, whereas the commutator of a symmetric and an anti-symmetric one is symmetric (this is the reason why symmetric matrices do not form a Lie algebra under commutation). Therefore, the commutator of non-commuting matrices of the form  $[i\id\otimes P_a,i\id\otimes P_a']$ gives a matrix of the same form $i\id\otimes P_a''$, up to real constant factors that can be  ignored. The non-zero commutators $[i\id\otimes P_a, i\{X,Y,Z\}\otimes P_s ]$  have the form $i\{X,Y,Z\}\otimes P_s'$. Besides, the commutator $[iX\otimes P_s, iX\otimes P_s']$ returns either zero or a matrix of the form $\{\id\otimes P_a\}$ (the same happens if we replace $X$ by $Y$ or $Z$ in the first qubit).
The commutator $[iX\otimes P_s, iZ\otimes P_s']$ gives $iY\otimes P_s
''$ or zero, which is true because we are commuting two symmetric matrices, and hence we must obtain an anti-symmetric one. And finally, $[iX\otimes P_s, iY\otimes P_s']$ results in $iZ\otimes P_s''$ or zero, and analogously under the exchange $X\leftrightarrow Z$ on the first qubit. Therefore, the set of operators in Proposition~\ref{prop:sp-algebra} is closed under commutation.

\section{Proof of Theorem 1}
\label{ap:th-1}

We here prove Theorem~\ref{th:symplectic-universal}, which reads as

\begin{theorem}
    The set of unitaries of the form in Eq.~\eqref{eq:circuit}, with  generators taken from
    \begin{equation}\label{eq-ap:generators}
        \GC = \{Y_i\}_{i=1}^{n} \cup \{X_iY_{i+1}, Y_iX_{i+1}\}_{i=2}^{n-1} \cup X_1 \cup Z_1 Z_2\,,
    \end{equation} 
is universal in $\mathbb{SP}(d/2)$, as
    \begin{equation}
    \Span_{\mathbb{R}} \langle i \GC\rangle_{\rm Lie}=  \mathfrak{sp}(d/2)\,.
\end{equation}
Here $\langle i\GC\rangle_{\rm Lie}$ is the Lie closure of $i\GC$, i.e., the set of operators obtained by the nested commutation of the elements in $i\GC$.
\end{theorem}

\begin{proof}
    We begin by showing that the generators $ \{Y_i\}_{i=2}^{n} \cup \{X_iY_{i+1}, Y_iX_{i+1}\}_{i=2}^{n-1}$  produce the full special orthogonal algebra $\so(d/2)$ in the last $n-1$ qubits. This algebra is the span (over the real numbers) of all the real anti-symmetric matrices. In other words, it is the span of all the Pauli strings with an odd number of $Y$'s (times $i$). Clearly, the generators $Y\otimes \id$, $\id\otimes Y$, $Y\otimes X$ and $X\otimes Y$ generate the full $\so(4)$ algebra, as it can be checked by direct calculation. From here, we proceed by induction. That is, we show that if we have the full $\so(2^k)$ algebra for some $k$, we obtain the $\so(2^{k+1})$ algebra by adding the operators  $\id_{2^k}\otimes Y$, $\id_{2^{k-1}}\otimes Y\otimes X$ and $\id_{2^{k-1}}\otimes X\otimes Y$, and taking commutators (see also~\cite{wiersema2023classification}). In particular, we notice that any anti-symmetric Pauli string with support on $k+1$ site takes the form $P_s\otimes Y$ or $P_a \otimes\{\id,X,Z\}$, where $P_s$ and $P_a$ are  arbitrary symmetric and anti-symmetric Pauli strings on $k$ sites, respectively. Since we already have the full $\so(2^k)$ algebra, we have all Pauli strings of the form $P_a \otimes \id$, where the support of $P_a$ on the $k$-th site can be any of $\{\id,X_k,Y_k,Z_k\}$. Commuting $\id_{2^{k-1}}\otimes X\otimes Y$ with all Pauli strings $P_a \otimes \id$ such that the support of $P_a$ on the $k$-th site is $Y_k$, produces all operators of the form $P_s\otimes Y$ such that the support of $P_s$ on the $k$-th site is $Z_k$. Further commuting the latter with $Y_k$ gives all the operators $P_s\otimes Y$ such that the support of $P_s$ on the $k$-th site is $X_k$. Likewise, commuting $\id_{2^{k-1}}\otimes X\otimes Y$ with all Pauli strings $P_a \otimes \id$ such that the support of $P_a$ on the $k$-th site is $Z_k$, produces all operators of the form $P_s\otimes Y$ such that the support of $P_s$ on the $k$-th site is $Y_k$. We now compute the commutators of $\id_{2^{k-1}}\otimes Y\otimes X$ with operators $P_s\otimes Y$ such that the support of $P_s$ on the $k$-th site is $Y_k$. This gives us all operators  $P_a\otimes Z$ such that the support of $P_a$ on the $k$-th site is $\id$, which upon commutation with $Y_k$ gives us all operators  $P_a\otimes X$ such that the support of $P_a$ on the $k$-th site is $\id$. Now, we know that the $\so(2^k)$ algebra contains all the anti-symmetric Paulis (times $i$). Hence, commuting operators from $\so(2^k)$ with $P_a\otimes \{X,Z\}$  such that the support of $P_a$ on the $k$-th site is $\id$, we can generate all operators of the form $P_a\otimes \{X,Z\}$. The last step is to obtain the operators $P_s\otimes Y$ such that the support of $P_s$ on the $k$-th site is $\id$. We achieve this via commutation of $\id_{2^{k-1}}\otimes Y\otimes X$ with operators of the form $P_a\otimes Z$ such that the support of $P_s$ on the $k$-th site is $Y$.

    We are left with the task of showing that adding the operators $X_1$ and $Z_1 Z_2$ indeed generates the $\spf(d/2)$ algebra, i.e., all the operators in Eq.~\eqref{eq-ap:sp-dla}. Those of the form $i\{\id\otimes P_a\}$ are already generated by the orthogonal operators in the last $n-1$ qubits. Then, starting with the commutators of $Z_1 Z_2$ with the operators $\id\otimes P_s$, we get all operators $Z\otimes P_s$ such that the support of $P_s$ on the second qubit is $X_2$ or $Y_2$. Further commuting those with  and $X_1$, $Y_1$ and $Y_2$ we obtain all operators of the form $\{X,Y,Z\}\otimes P_s$ such that the support of $P_s$ on the second qubit is not $\id$. To generate the remaining algebra operators, we commute $Z_1Z_2$ with the operators $\{X,Y\}\otimes P_s$ with support $Z_2$ on the second qubit, and then all the resulting operators with $X_1$. This way, we have generated all the operators in Eq.~\eqref{eq-ap:sp-dla}. Since the algebra is closed under commutation, and all our operators belong to it, this concludes the proof.
\end{proof}

\section{Proof of Proposition 2}
\label{ap:prop-2}

Here, we present the proof of Proposition \ref{prop:symplectic-local}, that we restate for convenience.

\begin{proposition} 
    The set of unitaries of the form in Eq.~\eqref{eq:circuit}, with  generators taken from
    \begin{equation}\label{eq-ap:generators-local}
        \GC_L = \bigcup_{i=1}^{n-1}\{X_i ,Y_i, Y_{i+1}, X_iX_{i+1}\}\,,
    \end{equation} 
is universal in $\mathbb{SU}(d)$, as
    \begin{equation}
    \Span_{\mathbb{R}} \langle i \GC_L\rangle_{\rm Lie}=  \mathfrak{su}(d)\,.
\end{equation}
\end{proposition}

\begin{proof}
    The first step is to notice that we can generate via commutators the full special unitary algebra $\su(d/2)$  in the first $n-1$ qubits. This is true because we have the single-qubit operators $X,Y,Z$ acting on all of the latter, and also the two-qubit operators $X\otimes X$ on all qubit pairs $\{(q,q+1)\}_{q=1}^{n-2}$. Commuting $X\otimes X$ with the single-qubit Paulis, we obtain all the two-qubit nearest-neighbors Pauli operators. From here, one can generate the entire $\su(d/2)$ algebra as detailed in~\cite{larocca2021diagnosing}. The last step is computing the commutators of operators in $\su(d/2)$ with the generators acting on the last two qubits. Here, we note that commuting with $Y_{n-1}$ just transforms $X_{n-1}$ into $ Z_{n-1}$  (up to a constant factor), and vice versa, and hence it does not generate any new linearly independent operator, since we already have the entire $\su(d/2)$ algebra. The same is true for commutators with $X_{n-1}$. Furthermore, all operators in $\su(d/2)$ have trivial support in the last qubit and therefore commute with $Y_n$.
    We then turn to the commutations with the two-qubit operators acting on the last pair of qubits, which are $\{\{X,Y,Z\}\otimes \{X,Z\}\}$ according to Eq.~\eqref{eq-ap:sp-dla}. Commuting these with operators in $\su(d/2)$ produce all operators of the form $P\otimes \{X_n, Z_n\}$, where $P\neq \id_{2^{n-1}}$ has support on the first $n-1$ qubits and $\{X_n, Z_n\}$ are the $X,Z$ Pauli matrices acting on the $n$-th qubit. Further commuting the latter, we can obtain all three-qubit operators of the form $P_{n-2}\otimes P_{n-1}'\otimes Y$, where $P_{n-2}$ and $P_{n-1}'$ are arbitrary Pauli operators (different from the identity) acting on the qubits $n-2$ and $n-1$. Finally, commuting e.g., $Y_{n-2} \otimes Y_{n-1}\otimes Y_n$ with $Y_{n-2} \otimes Y_{n-1}\otimes X_n$, we get $Z_n$, from which we can generate all single-qubit and two-qubit operators acting on the last two qubits. Since we know that these are universal for quantum computation, we have generated the full $\su(d)$ algebra.
\end{proof}

\section{Asymptotic Weingarten calculus}\label{ap:asympt}

We now discuss the reasons why the Gram matrix W in Eq.~\eqref{eq:gram-split} takes that form, namely
\begin{equation}\label{eq-ap:gram-split}
    W = d^t\left(\id_{D_t}+\frac{1}{d} B\right)\,,
\end{equation}
with the entries of $B\in\OC(1)$.
This follows from the fact that $\Tr[F_d(\sigma^T) F_d(\sigma)]=d^t$ $\forall \sigma\in\brauer$ and $\left|\Tr\left[F_d(\sigma^T)F_d(\pi)\right]\right|\leq d^{t-1}$  when $\pi\neq\sigma$. 
Diagrammatically, $\sigma$ and $\sigma^T$ are specular images of each other  $\forall \sigma \in   \brauer$. For the case of permutations in $S_t$, it is  then clear that $\Tr\left[F_d(\sigma^T)F_d(\sigma)\right]= \Tr\left[F_d(\sigma)\ad F_d(\sigma)\right]=\Tr\left[\id_d^{\otimes t}\right]=d^t$, as $\sigma^T=\sigma^{-1}$ and $F_d(\sigma^{-1})=F_d(\sigma)\ad $. 
For the rest of the elements in $\brauer$, which do not have an inverse, we notice that $\sigma$ and $\sigma^T$ are such that for every pair $\lambda_\gamma, \sigma(\lambda_\gamma)\leq t$ there exists a pair $\lambda_{\gamma'}, \sigma^T(\lambda_{\gamma'}) > t$ such that $\lambda_{\gamma'}=\lambda_\gamma+t $ and $\sigma^T(\lambda_{\gamma'})=\sigma(\lambda_\gamma)+t$. A similar result holds  for every pair $\lambda_\gamma, \sigma(\lambda_\gamma)> t$ (this is a consequence of $\sigma$ and $\sigma^T$ being specular images). Hence, the number of $-d$ factors (or equivalently, closed loops) that appear in $F_d(\sigma^T)F_d( \sigma)$ is even, and we have $\Tr[F_d(\sigma^T) F_d(\sigma)]=d^t$.
All other entries of the Gram matrix, $\Tr[F_d(\sigma^T)F_d(\pi)]$ where $\pi\neq\sigma$, are the same as those for the representation of the Brauer algebra $\mathfrak{B}_t(d)$ from the orthogonal Schur-Weyl duality, up to a minus sign in some entries. Therefore, it follows that $\left|\Tr\left[F_d(\sigma^T)F_d(\pi)\right]\right|\leq d^{t-1}$ (see e.g., Supplemental Proposition 13 in Ref.~\cite{garcia2023deep}), and so we find Eq.~\eqref{eq-ap:gram-split}.

\section{Proof of Theorem 2}
\label{ap:th-2}

In this Appendix we provide the proof of Theorem~\ref{th:gauss-1}, which states that the outputs of random symplectic quantum circuits converge in distribution to a Gaussian process under certain conditions.
    
\begin{theorem}\label{th-ap:gauss-1}
    Let $\mathscr{C}$ be a vector of expectation values of the Hermitian operator $O$ over a set of states from $\mathscr{D}$, as in Eq.~\eqref{eq:random_vector}. If  $\Tr[\rho_j \rho_{j'}]\in\Omega\left(1/\poly(\log(d))\right)$ and $\Big|\Tr[\Omega\rho_j \Omega\rho_{j'}^T]\Big|\in o\left(1/\poly(\log d)\right)$ $\forall j,j'$, then in the large $d$-limit $\mathscr{C}$ forms a GP with mean vector $\vec{\mu}=\vec{0}$ and covariance matrix
\begin{equation}\label{eq-ap:covariance-gp1}
    \vec{\Sigma}_{j,j'} = \frac{\Tr[\rho_{j}\rho_{j'}]}{d} \,.
\end{equation}
\end{theorem}

\begin{proof}

Our goal is to compute the moments of $\mathscr{C}$, that is, quantities of the form
\begin{equation}
    \mathbb{E}_{\SPBB(d/2)}\left[C(\rho_{j_1})\cdots C(\rho_{j_t})\right]\,.
\end{equation}
Using the linearity of the trace, and the fact that $\Tr[A] \Tr[B] = \Tr[A\otimes B]$, we find that
\small
\begin{equation}\begin{split}
    &\int_{\SPBB(d/2)} d\mu(U)\, \Tr[U\rho_{j_1}U^\dagger O ] \cdots \Tr[U\rho_{j_t}U^\dagger O ] = \\
    & \Tr\left[\int_{\SPBB(d/2)} d\mu(U)\, U^{\otimes t} \Lambda (U^\dagger)^{\otimes t} O^{\otimes t}\right] = \Tr\left[\TC^{(t)}_{\SPBB(d/2)}[\Lambda] \,O^{\otimes t}\right]\,,
\end{split}
\end{equation}
\normalsize
where we defined $\Lambda\equiv \rho_{j_1}\otimes\cdots\otimes \rho_{j_t}$. We first exactly compute the first and second moments. For $t=1$, we have
\begin{equation}
    \mathbb{E}_{\SPBB(d/2)}\left[C(\rho_j)\right] = \Tr\left[\TC^{(t)}_{\SPBB(d/2)}[\rho_j] \,O\right] = \Tr\left[\frac{\id_d}{d} \,O\right] = 0\,,\nonumber
\end{equation}
where we used Eq.~\eqref{eq:twirl-1} and the property that $O$ is traceless. The second-order twirl is given by Eq.~\eqref{eq:twirl-2}, i.e.,
\footnotesize
\begin{align} \label{eq-ap:twirl-2}
    \TC^{(2)}_{\mathbb{SP}(d/2)}\left[\Lambda\right]&=\frac{(d-1)\Tr[\Lambda]  -\Tr[\Lambda{\rm SWAP}]  +\Tr[\Lambda \Pi_s]}{d(d+1)(d-2)}\,\id_d\otimes \id_d\nonumber\\
    &+\frac{-\Tr[\Lambda] + (d-1)\Tr[\Lambda{\rm SWAP}]  -\Tr[\Lambda \Pi_s]}{d(d+1)(d-2)} \,{\rm SWAP}\nonumber\\
    &+\frac{\Tr[\Lambda]  -\Tr[\Lambda{\rm SWAP}] + (d-1) \Tr[\Lambda \Pi_s]}{d(d+1)(d-2)}\,\Pi_s\,.
\end{align}
\normalsize
Using Eqs.~\eqref{eq-ap:twirl-2} and~\eqref{eq:ricochet}, the covariance matrix entries are found to be
\begin{equation}\label{eq-ap:cov-exact}
\begin{split}
     \vec{\Sigma}_{j, j'}^{\mathbb{SP}}&=\frac{1}{d+1} \left(\Tr[\rho_j \rho_{j'}] +\Tr[\Omega\rho_j \Omega \rho_{j'}^T] \right) \,.\end{split}
\end{equation}
Let us evaluate $\Tr[\Omega P_j \Omega P_{j'}^T]$ for Pauli operators $P_j,P_{j'}$, which gives
\begin{equation} \label{eq:pauli-omega-trace}
    \Tr[iY\otimes\id_{d/2}\, P_j \,iY\otimes \id_{d/2} \,P_{j'}^T]=\begin{cases}
        d\, \delta_{j j'}\quad {\rm if\; P\in i\mathfrak{sp}(d/2)}\,, \\
        -d \,\delta_{j j'} \quad {\rm otherwise}\,.
    \end{cases}
\end{equation}
Since every quantum state can be written as
\begin{equation}
    \rho =\frac{1}{d} \sum c_k P_k\,,
\end{equation}
where the $c_k$ are real coefficients such that $-1\leq c_k \leq 1$ $\forall k$, it follows that
\begin{align} 
    \Tr[\Omega\rho_j \Omega \rho_j^T ] &= \frac{1}{d} \left(\sum_{k\backslash P_k\in i\mathfrak{sp}(d/2)} c_k^2 - \sum_{k\backslash P_k\notin i\mathfrak{sp}(d/2)} c_k^2 \right)\nonumber  \\ &= 2 \,\Tr_{\mathfrak{g}}[\rho_j^2] - \Tr[\rho_j^2]\,,
\end{align}
where we defined $\Tr_{\mathfrak{g}}\left[\cdot\right]$ as in Eq.~\eqref{eq:g-covariance}. Analogously,
\begin{equation} \label{eq-ap:bell-contribution}
    \Tr[\Omega\rho_j \Omega \rho_{j'}^T ] = 2 \,\Tr_{\mathfrak{g}}[\rho_j \rho_{j'}] - \Tr[\rho_j \rho_{j'}]\,.
\end{equation} 
Hence, we arrive at the following expression for the covariance matrix entries,
\begin{equation} \label{eq-ap:exact-cov}
    \vec{\Sigma}_{j, j'}^{\mathbb{SP}}=\frac{2 \,\Tr_{\mathfrak{g}}\left[\rho_j\rho_{j'}\right]}{d+1}\,.
\end{equation}
Using that $\Tr[\rho_j \rho_{j'}]\in\Omega\left(1/\poly(\log(d))\right)$ and $\Big|\Tr[\Omega\rho_j \Omega\rho_{j'}^T]\Big|\in o\left(1/\poly(\log d)\right)$ $\forall j,j'$, we can approximate this covariance in the large-$d$ limit as in Eq.~\eqref{eq-ap:covariance-gp1}.
We remark here that we could have substituted any Pauli $P_j$ by $O'=S P_j S^\dagger$, where $S\in \SPBB(d/2)$, in Eq.~\eqref{eq:pauli-omega-trace}, and that equation would still hold. This follows from the symplectic condition $S^T \Omega S = \Omega$ and the fact that the transpose of a symplectic matrix is symplectic. This implies that our GP results will also be valid when we replace $O$ by $O'$.

To compute higher moments we will use the asymptotic Weingarten calculus for the symplectic group explained in Sec.~\ref{sec:weingarten}. In particular, Eq.~\eqref{eq:twirl-t} gives
\footnotesize
\begin{equation} \label{eq-ap:gp-twirl} \begin{split}
    \Tr\left[\TC^{(t)}_{\SPBB(d/2)}[\Lambda] \,O^{\otimes t}\right] 
   &=\frac{1}{d^t}\sum_{\sigma\in \brauer}\Tr\left[F_d(\sigma^T)\Lambda \right]\Tr\left[F_d(\sigma) O^{\otimes t}\right]\\&+\!\!\!\!\sum_{\sigma,\pi\in \brauer}\!\!\!\!\frac{c_{\pi,\sigma}}{d^t}\,\Tr\left[F_d(\sigma^T)\Lambda\right]\Tr\left[F_d(\pi) O^{\otimes t}\right].
\end{split}
\end{equation}
\normalsize
Let us first focus on the factors of the form $\Tr\left[F_d(\sigma) O^{\otimes t}\right]$. It is straightforward to show that $\Tr\left[\sigma O^{\otimes t}\right]=0$ whenever $\sigma$ contains a cycle of odd length (this is a consequence of $O$ being traceless), and $\Big|\Tr\left[F_d(\sigma) O^{\otimes t}\right]\Big|=d^{|\sigma|}$, where $|\sigma|$ is the number of cycles in $\sigma$, otherwise. We refer to Supplemental Proposition 6 in Ref.~\cite{garcia2023deep} for a detailed derivation of the previous when $\sigma$ is a permutation. 

When $\sigma$ is not a permutation, we simply note that for all $\sigma$
\begin{equation}
    \Tr\left[F_d(\sigma) O^{\otimes t}\right] = \pm \prod_{\alpha=1}^{|\sigma|} \Tr\left[O^{|c_\alpha|} \Omega^{2 b_\alpha}\right]=\pm d^{|\sigma|}\,,
\end{equation}
where the product runs over the cycles in $\sigma$, $|c_\alpha|$ is the length of the cycle $c_\alpha$, and $b_\alpha$ is the number of pairs $(i,c_\alpha(i))$ in $c_\alpha$ such that both $i$ and $c_\alpha(i)\leq t$ (i.e., the number of pairs $(i,c_\alpha(i))$ that are in the first column). Here, we used the fact that $\Omega$ either commutes or anti-commutes with $O$, together with $O^T=\pm O$ and $\Omega^T = -\Omega=\Omega^{-1}$.
It is clear then that  $\Big|\Tr\left[F_d(\sigma) O^{\otimes t}\right]=d^{|\sigma|}\Big|$ is maximized whenever $\sigma$ consists of a product of disjoint length-two cycles.  Otherwise it is at least $\frac{1}{d}$ times smaller. Moreover, if $\sigma$ is  a product of disjoint length-two cycles, then $\Tr\left[F_d(\sigma) O^{\otimes t}\right]=d^{t/2}$ since $\Tr[\Omega O\Omega O^T]=d$ when $O$ is a Pauli in $i\spf(d/2)$.

We then need to study the terms of the form $\Tr\left[F_d(\sigma^T)\Lambda\right]$. When $\sigma$ is a permutation, it holds that $\Big|\Tr\left[F_d(\sigma^T)\Lambda\right]\Big| \leq 1$. (see Supplemental Proposition 7 in~\cite{garcia2023deep}). Moreover this is also true for the non-permutation elements in $\brauer$. To see this, it suffices to notice that when one expands $\Big|\Tr\left[F_d(\sigma^T)\Lambda \right]\Big|$ there will appear quantities of the form $\bra{\psi_i}\Omega\ket{\psi_{i'}}$ in addition to those of the form $\bra{\psi_i}\psi_{i'}\rangle$ in the products. Since $\Omega\ket{\psi_{i'}}$ is just another quantum state $\ket{\widetilde{\psi_{i'}}}$, the result follows.
This implies that $\Tr\left[F_d(\sigma^T)\Lambda \right]$ cannot be large so as to compensate the factors $\OC(1/d)$ by which $\Big|\Tr\left[F_d(\sigma) O^{\otimes t}\right]\Big|$ are suppressed when $\sigma$ is not the disjoint product of length-two cycles. They could however be very small in principle, and that is why we require $\Tr[\rho_j \rho_{j'}]\in\Omega\left(1/\poly(\log(d))\right)$ $\forall j,j'$, which implies that $\Tr\left[F_d(\sigma^T)\Lambda \right]\in\Omega\left(1/\poly(\log(d))\right)$ when $\sigma^T$ is a disjoint product of length-two cycles.

We now recall that the permutations such that $\sigma=\sigma^{T}$ are known as involutions and must  consist of a product of disjoint transpositions plus fixed points. More generally, we have that for an element $\sigma\in\brauer$ to satisfy that $\sigma=\sigma^T$ it must consist of a product of disjoint length-two cycles and fixed points. We denote as $T_t$ the set of permutations in $\sigma\in\mathbb{S}_t$ that are a product of disjoint length-two cycles. Therefore, using~\eqref{eq-ap:gp-twirl} and the fact that $\Big|\Tr[\Omega\rho_j \Omega\rho_{j'}^T]\Big|\in o\left(1/\poly(\log d)\right)$ $\forall j,j'$ (this condition allows us to only retain the contribution from the elements in $T_t$, instead of all the $\sigma\in\brauer$ that are the product of disjoint length-two cycles), we arrive at
\small
\begin{equation} \label{eq-ap:asym-twirl}
    \Tr\left[\TC^{(t)}_{\SPBB(d/2)}[\Lambda] \,O^{\otimes t}\right] 
   =\frac{1}{d^{t/2}}\sum_{\sigma\in T_t}\Tr\left[F_d(\sigma^T)\Lambda \right] + \OC\left(\frac{1}{d^{t/2+1}}\right) \,.
\end{equation}
\normalsize
Or just retaining the leading-order terms,
\begin{equation} \label{eq-ap:gp-asym-twirl}
    \Tr\left[\TC^{(t)}_{\SPBB(d/2)}[\Lambda] \,O^{\otimes t}\right]\approx \frac{1}{d^{t/2}}\sum_{\sigma\in T_t} \prod_{(j,j')\in\sigma} \Tr\left[\rho_j \rho_{j'}\right]\,,
\end{equation}
where the product runs over all the cycles $(j,j')$ in $\sigma$.

The last step is comparing Eq.~\eqref{eq-ap:gp-asym-twirl} with the moments of a multivariate Gaussian, which are given by Wick's or Isserlis' theorem~\cite{isserlis1918formula}. This theorem states that, for random variables $X_{1},X_{2},\dots, X_{t}$ that form a GP, the $t$-th order moment  is $\mathbb{E}[X_{1}X_{2}\cdots X_{t}]=0$ if $t$ is odd, and
\begin{equation}\label{eq-ap:wick}
    \mathbb{E}[X_{1}X_{2}\cdots X_{t}]=
        \sum_{\sigma\in T_{t}}\prod_{(j,j')\in \sigma} {\rm Cov} [X_{j},X_{j'}]\,,
\end{equation}
if $t$ is even. 
Clearly, Eq.~\eqref{eq-ap:gp-asym-twirl} matches Eq.~\eqref{eq-ap:wick} by identifying ${\rm Cov} [X_{j},X_{j'}]=\frac{\Tr[\rho_j \rho_{j'}]}{d}$ as in Eq.~\eqref{eq-ap:covariance-gp1}.
Finally, it can be proven that these moments uniquely determine the distribution of $\mathscr{C}$ using Carleman's condition (see Ref.~\cite{garcia2023deep}). Hence, $\mathscr{C}$ forms a GP.

\end{proof}

\section{Proof of Theorem 3}
\label{ap:th-3}

In this appendix we prove Theorem~\ref{th:gauss-2}, which is a more general statement about the convergence to GPs of ramdom symplectic quantum circuits than that from Theorem~\ref{th:gauss-1}.

\begin{theorem}\label{th-ap:gauss-2}
    Let $\mathscr{C}$ be a vector of expectation values of the Hermitian operator $O$ over a set of states from $\mathscr{D}$, as in Eq.~\eqref{eq:random_vector}. If  $\Big|\Tr[\rho_j \rho_{j'}]+\Tr[\Omega\rho_j \Omega\rho_{j'}^T]\Big|\in\Omega\left(1/\poly(\log(d))\right)$  $\forall j,j'$, then in the large $d$-limit $\mathscr{C}$ forms a GP with mean vector $\vec{\mu}=\vec{0}$ and covariance matrix
\begin{equation}\label{eq-ap:covariance-gp2}
    \vec{\Sigma}_{j,j'} = \frac{2\,\Tr_{\mathfrak{g}}[\rho_{j}\rho_{j'}]}{d} \,.
\end{equation}
\end{theorem}

\begin{proof}
    The proof of this result is completely analogous to that of Theorem~\ref{th-ap:gauss-1}. The only difference is that we now need to retain all the contributions coming from the set of elements of the Brauer algebra consisting of disjoint products of length-two cycles (that we denote $\mathfrak{T}_t$), and not just  those arising from the elements in $T_t$. This is so because of the condition $\Big|\Tr[\rho_j \rho_{j'}]+\Tr[\Omega\rho_j \Omega\rho_{j'}^T]\Big|\in\Omega\left(1/\poly(\log(d))\right)$  $\forall j,j'$. 
    
    Hence, instead of Eq.~\eqref{eq-ap:asym-twirl}, we find
    \footnotesize
    \begin{equation} 
    \Tr\left[\TC^{(t)}_{\SPBB(d/2)}[\Lambda] \,O^{\otimes t}\right] 
   =\frac{1}{d^{t/2}}\sum_{\sigma\in \mathfrak{T}_t}\Tr\left[F_d(\sigma^T)\Lambda \right] + \OC\left(\frac{1}{d^{t/2+1}}\right) \,.
\end{equation}
\normalsize
Now, we can group the contributions in the previous equation as follows. Let us consider a permutation  $\sigma\in T_t$. For each pair $(j,j')\in \sigma$, we can substitute the transposition by a two-length cycle that is not a permutation, obtaining an element in $\mathfrak{T}_t$. This amounts to replacing a factor $\Tr\left[\rho_j\rho_{j'}\right]$ by a factor $\Tr[\Omega\rho_j \Omega\rho_{j'}^T]$ in $\Tr\left[F_d(\sigma^T)\Lambda \right]$. 
We can do this for every possible pair $(j,j')$, and thus, we have that
\footnotesize
\begin{equation}
    \sum_{\sigma\in \mathfrak{T}_t}\Tr\left[F_d(\sigma^T)\Lambda \right] = \sum_{\sigma\in T_t} \prod_{(j,j')\in\sigma} \left(\Tr[\rho_j \rho_{j'}]+\Tr[\Omega\rho_j \Omega\rho_{j'}^T]\right)\,.
\end{equation}
\normalsize
Therefore, using~\eqref{eq-ap:bell-contribution} we arrive at
\begin{equation} \label{eq-ap:gp2-asym-twirl}
    \Tr\left[\TC^{(t)}_{\SPBB(d/2)}[\Lambda] \,O^{\otimes t}\right] \approx \frac{1}{d^{t/2}} \sum_{\sigma\in T_t} \prod_{(j,j')\in\sigma} 2\,\Tr_{\mathfrak{g}}[\rho_{j}\rho_{j'}]\,.
\end{equation}
Again, comparing Eqs.~\eqref{eq-ap:gp2-asym-twirl} and~\eqref{eq-ap:wick} it is clear that they match by identifying ${\rm Cov} [X_{j},X_{j'}]=\frac{2\Tr_{\g}[\rho_j \rho_{j'}]}{d}$ as in Eq.~\eqref{eq-ap:covariance-gp2} (note that we approximated $d+1\approx d$ for large $d$, see Eq.~\eqref{eq-ap:exact-cov}). We conclude that $\mathscr{C}$ forms a GP.

\end{proof}

\section{Proof of Theorem 4}
\label{ap:th-4}

We here present the proof of Theorem~\ref{th:gauss-3}, which shows that random symplectic quantum circuits can form uncorrelated GPs.

\begin{theorem}
    Let $\mathscr{C}$ be a vector of expectation values of the Hermitian operator $O$ over a set of states from $\mathscr{D}$, as in Eq.~\eqref{eq:random_vector}. If  $\Tr[\rho_j^2]+\Tr[\Omega\rho_j \Omega\rho_j^T]\in\Omega\left(1/\poly(\log(d))\right)$  $\forall j$ and $\Tr[\rho_j \rho_{j'}]=-\Tr[\Omega\rho_j \Omega\rho_{j'}^T]$ $\forall j\neq j'$, then in the large $d$-limit $\mathscr{C}$ forms a GP with mean vector $\vec{\mu}=\vec{0}$ and diagonal covariance matrix
\begin{equation}\label{eq-ap:covariance-gp3}
    \vec{\Sigma}_{j,j'} = \begin{cases}
        \frac{2\,\Tr_{\mathfrak{g}}[\rho_{j}^2]}{d} \quad {\rm if\; } j=j' \\
         \; 0 \qquad \quad\quad\;\; {\rm if\; } j\neq j'
    \end{cases} \,.
\end{equation}
\end{theorem}

\begin{proof}

To prove this result we will use a somewhat different strategy than that employed in the proofs of Theorems~\ref{th:gauss-1} and~\ref{th:gauss-2}. Namely, we will leverage the fact the if $\mathscr{C}$ is a GP with zero mean (i.e., $\vec{\mu}=\vec{0}$), then any linear combination of its entries, $\LC=\sum_j  a_j C(\rho_j)$, where the $a_j$ are constants,  follows a Gaussian distribution $\NC(0,\sigma^2)$ with $\sigma^2=\sum_{j,j'} a_j a_{j'}{\rm Cov}[C(\rho_j), C(\rho_{j'})]$.

Let us then compute the moments of $\LC$. For the first moment, we have
\begin{equation}
    \mathbb{E}_{\SPBB(d/2)}[\LC] = \sum_j  a_j\,  \mathbb{E}_{\SPBB(d/2)}[C(\rho_j)]=0\,.
\end{equation}
For the second moment,
\begin{align}\label{eq-ap:var-uncorrelated}
    \mathbb{E}_{\SPBB(d/2)}\left[\LC^2\right] \nonumber&= \sum_{j,j'} a_j a_{j'} \, \mathbb{E}_{\SPBB(d/2)}\left[C(\rho_j), C(\rho_{j'})\right]\\&= \sum_j a_j^2 \, \mathbb{E}_{\SPBB(d/2)}\left[C(\rho_j)^2\right] \nonumber \\&= \sum_j a_j^2\,  \frac{2\,\Tr_{\mathfrak{g}}\left[\rho_{j}^2\right]}{d}\,,
\end{align}
where we used Eq.~\eqref{eq-ap:cov-exact} together with the condition that  $\Tr[\rho_j \rho_{j'}]=-\Tr[\Omega\rho_j \Omega\rho_{j'}^T]$ $\forall j\neq j'$.
It is clear that Eq.~\eqref{eq-ap:var-uncorrelated} matches $\sigma^2=\sum_{j,j'} a_j a_{j'}{\rm Cov}[C(\rho_j), C(\rho_{j'})]$, when the covariance matrix entries are given as in Eq.~\eqref{eq-ap:covariance-gp3}.  

Then, we find that every higher-order moment takes the form 
\footnotesize
\begin{align}
    &\mathbb{E}_{\SPBB(d/2)}\left[\LC^t\right] = \nonumber\\&\sum_{t_1+\dots+t_m=t} \binom{t}{t_1,\dots, t_m}  \mathbb{E}_{\SPBB(d/2)}\left[a_1^{t_1}C(\rho_1)^{t_1}\cdots a_m^{t_m} C(\rho_m)^{t_m}\right]\,,
\end{align}
\normalsize
where all the $t_j$ are non-negative. When $t$ is odd, it follows that $\mathbb{E}_{\SPBB(d/2)}\left[a_1^{t_1}C(\rho_1)^{t_1}\cdots a_m^{t_m} C(\rho_m)^{t_m}\right]=0$ from the fact that $\Tr\left[F_d(\sigma) O^{\otimes t}\right]$ $\forall \sigma\in\brauer$, as explained in Appendix~\ref{ap:th-2}. This matches the odd moments of a Gaussian distribution. When $t$ is even, we begin by proving that given two set of values $\{t_1,\dots,t_m\}$ and $\{t'_1,\dots,t'_m\}$ 
\begin{equation} \label{eq-ap:dominant-moments}
     \frac{\mathbb{E}_{\SPBB(d/2)}\left[C(\rho_1)^{t_1}\cdots C(\rho_m)^{t_m}\right]}{\mathbb{E}_{\SPBB(d/2)}\left[C(\rho_1)^{t'_1}\cdots C(\rho_m)^{t'_m}\right]}\in o(1)\,,
\end{equation}
whenever there exists an odd number in $\{t_1,\dots,t_m\}$ but all numbers in $\{t'_1,\dots,t'_m\}$  are even. To see that this is case, we simply note that if there exists an odd number in $\{t_1,\dots,t_m\}$, then there is at least one length-two cycle in any element in $\mathfrak{T}_t$ that connects two different states $\rho_j,\rho_{j'}$ (with $j\neq j'$) in $\Tr\left[F_d(\sigma^T)\Lambda \right]$. For each such pair $(j,j')\in \sigma^T$, we can either have a transposition or a two-length cycle that is not a permutation (which we recall produce a factor $\Tr\left[\rho_j\rho_{j'}\right]$ or $\Tr[\Omega\rho_j \Omega\rho_{j'}^T]$, respectively).
Hence, since $\Tr[\rho_j \rho_{j'}]=-\Tr[\Omega\rho_j \Omega\rho_{j'}^T]$ $\forall j\neq j'$ we find that the sum of all contributions to Eq.~\eqref{eq-ap:gp-twirl} coming from elements in $\mathfrak{T}_t$ is zero, and so we need to take into account elements $\sigma\in \brauer$ that are not the product of disjoint length-two cycles in order to compute the larger non-zero contributions to $\mathbb{E}_{\SPBB(d/2)}\left[C(\rho_1)^{t_1}\cdots C(\rho_m)^{t_m}\right]$. However, we already know from Appendices~\ref{ap:th-2} and~\ref{ap:th-3} that these are suppressed as $\OC(1/d)$ compared to the values $\mathbb{E}_{\SPBB(d/2)}\left[C(\rho_1)^{t'_1}\cdots C(\rho_m)^{t'_m}\right]$ when all the numbers in $\{t'_1,\dots,t'_m\}$  are even. Thus, Eq.~\eqref{eq-ap:dominant-moments} holds.

Following these considerations, we find that the leading-order contributions give us (for even $t$)
\small
\begin{align} \label{eq-ap:uncorr-moments}
    &\mathbb{E}_{\SPBB(d/2)}\left[\LC^t\right] \approx \nonumber\\&\sum_{\substack{t_1+\dots+t_m=t \\ t_1,\dots,t_m\; {\rm even}}} \binom{t}{t_1,\dots, t_m}  \mathbb{E}_{\SPBB(d/2)}\left[a_1^{t_1}C(\rho_1)^{t_1}\cdots a_m^{t_m} C(\rho_m)^{t_m}\right] \nonumber \\& = \sum_{\substack{t_1+\dots+t_m=t \\ t_1,\dots,t_m\; {\rm even}}} \binom{t}{t_1,\dots, t_m}  \prod_{j=1}^m (t_j-1)!! \,a_j^{t_j} \left(\frac{2\,\Tr_{\mathfrak{g}}\left[\rho_{j}^2\right]}{d}\right)^{t_j/2} \nonumber \\&=\sum_{\substack{t_1+\dots+t_m=t \\ t_1,\dots,t_m\; {\rm even}}}\frac{t!}{\prod_{j=1}^m t_j!!}  \prod_{j=1}^m \,a_j^{t_j} \left(\frac{2\,\Tr_{\mathfrak{g}}\left[\rho_{j}^2\right]}{d}\right)^{t_j/2}  \,.
\end{align}
\normalsize
where we used that $t_j!=t_j!!(t_j-1)!!$.

We need to show that $\mathbb{E}_{\SPBB(d/2)}\left[\LC^t\right]$ matches the even moments of a variable $X$ following a zero-mean Gaussian distribution, which are given by
\begin{equation} \label{eq-ap:uncorr-moments-Gaussian}
    \mathbb{E}\left[X^t\right] = (t-1)!!\, \sigma^{t}\,.
\end{equation}
To do this, we begin by computing 
\small
\begin{align} \label{eq-ap:uncorr-power-var}
    \sigma^t &= \left(\sum_{j,j'} a_j a_{j'}{\rm Cov}[C(\rho_j), C(\rho_{j'})]\right)^{t/2} \nonumber \\ &= \left(\sum_{j} a_j^2\, \frac{2\,\Tr_{\mathfrak{g}}\left[\rho_{j}^2\right]}{d}\right)^{t/2} \nonumber\\&= \sum_{k_1+\dots+k_m=t/2} \binom{t/2}{k_1,\dots, k_m} \prod_{j=1}^m a_j^{2k_j}\left(\frac{2\,\Tr_{\mathfrak{g}}\left[\rho_{j}^2\right]}{d}\right)^{k_j}   \,,
\end{align}
\normalsize
where we used Eq.~\eqref{eq-ap:var-uncorrelated}. Now, since the sum in Eq.~\eqref{eq-ap:uncorr-moments} is restricted to even values $t_1,\dots,t_m$, we can re-write it by re-labeling $t_j=2k_j$, obtaining
\begin{align} \label{eq-ap:uncorr-final}
   & \mathbb{E}_{\SPBB(d/2)}\left[\LC^t\right] \approx \nonumber \\& \sum_{2 k_1+\dots+2k_m=t }\frac{t!}{\prod_{j=1}^m (2k_j)!!}  \prod_{j=1}^m \,a_j^{2k_j} \left(\frac{2\,\Tr_{\mathfrak{g}}\left[\rho_{j}^2\right]}{d}\right)^{k_j} .
\end{align}
Since $t_j=2k_j$ is a bijective function, there are exactly the same number of terms in the sums in Eq.~\eqref{eq-ap:uncorr-moments} and~\eqref{eq-ap:uncorr-power-var}. To show that Eq.~\eqref{eq-ap:uncorr-final} is indeed equal to Eq.~\eqref{eq-ap:uncorr-moments-Gaussian} for all possible values of the constants $a_j$, we need to prove that 
\begin{align}
    (t-1)!! \binom{t/2}{k_1,\dots, k_m} = \frac{t!}{\prod_{j=1}^m (2k_j)!!} \,,
\end{align}
which follows from the properties of the double factorial for even and odd numbers. Namely, $(2k_j)!!=2^{k_j} k_j!$ and $(t-1)!!=\frac{t!}{2^{t/2}(t/2)!}$, which imply
\begin{align}
    \frac{t!}{2^{t/2}(t/2)!} \frac{(t/2)!}{\prod_{j=1}^m k_j!} = \frac{t!}{2^{t/2}\prod_{j=1}^m k_j!} \,.
\end{align}
Hence,  $\mathscr{C}$ forms a GP.
    
\end{proof}

\section{Proof of Corollary 1}
\label{ap:cor-1}

Let us now prove Corollary~\ref{cor:double}, which reads as follows. 

\begin{corollary} 
    Let $C(\rho_j)$ be the expectation value of a Haar random symplectic quantum circuit as in Eq.~\eqref{eq:ci}. If the conditions under which Theorems~\ref{th:gauss-1}, ~\ref{th:gauss-2} and~\ref{th:gauss-3} hold are satisfied, then 
    \begin{equation} \label{eq-ap:double-exp}
      {\rm Pr}(|C(\rho_j)|\geq c)\in\OC\left(\frac{\Tr_\g\left[\rho_j^2\right]}{c\sqrt{d} e^{dc^2 / (4\,\Tr_\g[\rho_j^2])}}\right)\,.
    \end{equation}
\end{corollary}

\begin{proof}
    Since the marginal distributions of a GP are Gaussian, $C(\rho_j)$ follows a  Gaussian distribution $\NC(0,\sigma^2)$ with $\sigma^2=\frac{2\,\Tr_\g\left[\rho_j^2\right]}{d}$ (see Eq.~\eqref{eq-ap:exact-cov}). Thus, we have that
\begin{equation}\label{eq-ap:tail-probab} \begin{split}
    {\rm Pr}(|C_j(\rho_i)|\geq c)&=\frac{\sqrt{d}}{\sqrt{\pi\,\Tr_\g\left[\rho_j^2\right]}}\int_{c}^{\infty} dx e^{-\frac{x^2 d}{4 \Tr_\g[\rho_j^2]}}\\ &={\rm Erfc}\left[\frac{c\sqrt{d}}{2 \,\Tr_\g\left[\rho_j^2\right]}\right]\,,
\end{split}
\end{equation}
where ${\rm Erfc}$ is the complementary error function. Using the fact that for large $x$, ${\rm Erfc}[x]\leq \frac{e^{-x^2}}{x\sqrt{\pi}}$, we find Eq.~\eqref{eq-ap:double-exp}.

\end{proof}

\section{Proof of Corollary 2}
\label{ap:cor-2}

In this section we present the proof for Corollary~\ref{cor:t-designs}, which reads

\begin{corollary}
Let $C(\rho_j)$ be the expectation value of a quantum circuit that forms a $t$-design over $\SPBB(d/2)$ as in Eq.~\eqref{eq:ci}. If the conditions under which Theorems~\ref{th:gauss-1}, ~\ref{th:gauss-2} and~\ref{th:gauss-3} hold are satisfied, then 
\small
\begin{equation} \label{eq-ap:concentration}
    {\rm Pr}(|C(\rho_j)|\geq c)\in\OC\left(\left(2\left\lfloor t/2\right\rfloor-1\right)!!\, \left(\frac{2\,\Tr_\g\left[\rho_j^2\right]}{d c^2}\right)^{\left\lfloor t/2\right\rfloor }\right)\,.
\end{equation}
\normalsize
\end{corollary}

\begin{proof}

Here we can use the generalization of Chebyshev's inequality to higher-order moments,
\begin{equation}
    \Pr \left(|X- \mathbb{E}[X]|\geq c\right)\leq \frac{  \mathbb{E}[|X-\mathbb{E} [X]|^t]}{c^{t}}\,,
\end{equation}
for $c>0$ and for $t\geq 2$. Since $\mathbb{E}_{\SPBB(d/2)}[C(\rho_j)] = 0$, this inequality simplifies to 
\begin{equation}\label{eq-ap:cheb-higher}
    \Pr \left(|C(\rho_j)|\geq c\right)\leq \frac{  \mathbb{E}_{\SPBB(d/2)}\left[|C(\rho_j)|^t\right]}{c^{t}}\,.
\end{equation}
If an ensemble of quantum circuits forms a $t$-design over $\SPBB(d/2)$, then we can readily evaluate $\mathbb{E}[C(\rho_j)^{t}]$, as this quantity matches the $t$ first moment of a $\NC(0,\sigma^2)$ distribution with $\sigma^2=\frac{2\,\Tr_\g\left[\rho_j^2\right]}{d}$. In particular, we know that since the odd moments are zero, we only need to take into account the largest even moment that the $t$-design matches. Therefore, using
\begin{equation} \begin{split}
    \mathbb{E}\left[|C_j(\rho_i)^{2\lfloor t/2\rfloor }\right]&=\left(2\left\lfloor t/2\right\rfloor-1\right)!!\; \mathbb{E}\left[C_j(\rho_i)^{2}\right]^{\left\lfloor t/2\right\rfloor }\\&=\left(2\left\lfloor t/2\right\rfloor-1\right)!!\; \left(\frac{2\,\Tr_\g\left[\rho_j^2\right]}{d}\right)^{\left\lfloor t/2\right\rfloor }\,,
\end{split}
\end{equation}
and plugging it into Eq.~\eqref{eq-ap:cheb-higher}, we arrive at~\eqref{eq-ap:concentration}.

\end{proof}

\section{Proof of Theorem 5}
\label{ap:th-5}

In this appendix we prove that sufficiently-random symplectic quantum circuits anti-concentrate, as stated in Theorem~\ref{th:anti-con}. We recall this result for convenience.

\begin{theorem}
Let $\mu$ be the Haar measure on $\SPBB(d/2)$, or a measure giving rise to a $2$-design over $\SPBB(d/2)$. Then,
    \begin{equation}
        {\rm Pr}_{U\sim \mu}\left(\left|\bra{x} U \ket{0}^{\otimes n}\right|^2 \geq \frac{\alpha}{d}\right) \geq \frac{(1-\alpha)^2}{2} \,,
    \end{equation}
for $0\leq \alpha\leq1$.
\end{theorem}

\begin{proof} 

In order to show anti-concentration, we will use Paley–Zygmund inequality~\cite{dalzell2022randomquantum,oszmaniec2022fermion}, which states that for a random variable $\mathscr{Z}\geq0$ with finite variance, it holds that
\begin{equation}
    {\rm Pr}(\mathscr{\mathscr{Z}}> c\, \mathbb{E}[\mathscr{Z}]) \geq (1-c)^2\frac{\mathbb{E}[\mathscr{Z}]^2}{\mathbb{E}[\mathscr{Z}^2]} \,,
\end{equation}
where $0\leq c\leq 1$. In our case, $\mathscr{Z}\equiv \Tr[U\rho_0 U^\dagger \Pi_x]$ with $\rho_0 =( \ketbra{0}{0})^{\otimes n}$ and $\Pi_x = \ketbra{x}{x} $. Hence, we need to compute
\small
\begin{equation} \begin{split}
    \mathbb{E}_{\SPBB(d/2)}\left[\Tr[U\rho_0 U^\dagger \Pi_x]\right] &= \Tr\left[\Pi_x\int_{\SPBB(d/2)} d\mu(U) U\rho_0 U^\dagger \right] \\  & = \Tr\left[\Pi_x\frac{\id_d}{d} \right] = \frac{1}{d}
    \end{split}
\end{equation}
\normalsize
where we used Eq.~\eqref{eq:twirl-1}. We also need to compute
\small
    \begin{equation} \begin{split}
    &\mathbb{E}_{\SPBB(d/2)}\left[\Tr[U\rho_0 U^\dagger \Pi_x]^2\right]\\
     &= \Tr\left[\Pi_x^{\otimes 2}\int_{\SPBB(d/2)} d\mu(U) U^{\otimes 2}\rho_0^{\otimes 2} (U^\dagger)^{\otimes 2} \right]\\
    &=  \Tr\left[\Pi_x^{\otimes 2}\left(\frac{\id_d\otimes \id_d+{\rm SWAP} }{d(d+1)}  \right)\right]\\
    & =  \frac{2}{d(d+1)}=\frac{z_H}{d}  \,,
    \end{split}
\end{equation}
\normalsize
where we used Eq.~\eqref{eq:twirl-2}, together with $\Tr\left[\rho_0^{\otimes 2}\right]=1$, $ \Tr\left[\rho_0^{\otimes 2} {\rm SWAP}\right] =\Tr\left[\rho_0^2\right]=1$, and
\footnotesize
    \begin{equation} 
    \begin{split}
    &\Tr\left[\rho_0^{\otimes 2} \Pi_s\right] \\
    &= \Tr\left[ \left(\id_d \otimes \Omega \rho_0 \Omega \rho_0^T\right) \Pi\right]\\
    &= \Tr\left[ \left(\id_d \otimes\Omega \rho_0 \Omega \rho_0\right) \Pi\right] \\
    &= \frac{1}{d^2} \Tr\left[ \left(\id_d \otimes \Omega \Big(\sum_{P_i\in\{\id,Z\}^{\otimes n}} P_i\Big) \Omega \Big(\sum_{P_{i'}\in\{\id,Z\}^{\otimes n}} P_i'\Big)\right)\Pi \right] \\
    &= \frac{1}{d^2} \Tr\left[ \left(\id_d \otimes \sum_{P_i\in\{\id,Z\}^{\otimes n}}\Omega  P_i\Omega  P_i\Big)\right)\Pi \right]  \\ 
    & =0\,.
\end{split}
\end{equation}
\normalsize
Hence, we arrive at 
\begin{equation}
    {\rm Pr}_{U\sim \mu}\left(\left|\bra{x} U \ket{0}^{\otimes n}\right|^2>\frac{\alpha}{d}\right) \geq (1-\alpha)^2\, \frac{d(d+1)}{2d^2}\,,
\end{equation}
which implies
\begin{equation}
    {\rm Pr}_{U\sim \mu}\left(\left|\bra{x} U \ket{0}^{\otimes n}\right|^2\geq\frac{\alpha}{d}\right) \geq \frac{(1-\alpha)^2}{2} \,.
\end{equation}
    
\end{proof}

\section{Tensor network-based calculation of the moments of shallow random  $\mathbb{SP}(2)$ circuits}

Let us recall that in~\cite{braccia2024computing} the authors present a toolbox to compute expectation values of circuits composed of local random gates via tensor networks. As mentioned in the main text, the key idea developed therein is to map the evaluation of $\widehat{\TC}^{(t)}_{\SC}$ in Eq.~\eqref{eq:moment-product} to a  Markov-chain like process where the $\widehat{\TC}^{(t)}_{G_l}$ only act on their local commutants. As such, in order to use such formalism, we need to a tensor representation for the superoperators $\widehat{\TC}_{\mathbb{SP}(2)}^{(2)}=\int_{\mathbb{SP}(2)} d\mu(V) V^{\otimes 2}\otimes (V^*)^{\otimes 2}$ and $\widehat{\TC}_{\mathbb{SO}(4)}^{(2)}=\int_{\mathbb{SO}(4)} d\mu(V) V^{\otimes 2}\otimes (V^*)^{\otimes 2}$.

We have shown in the main text that a basis for $\CC^{(2)}_{\SPBB(2)}$ is $\{\id_d\otimes\id_d, {\rm SWAP}, \Pi_s\}$, where we recall that these operators act on two copies of the two qubits targeted by a $\mathbb{SP}(2)$ gate. We will refer to these two two-qubit systems as $A$ and $B$.
With this in mind, we note that 
\small
\begin{equation}
    {\rm SWAP} = {\rm SWAP}_A \otimes {\rm SWAP}_B = \frac{\id_A + {\rm S}_A}{2} \otimes \frac{\id_B + {\rm S}_B}{2}\, ,
\end{equation}
\normalsize
where $S_{J}=X_{J_1}\otimes X_{J_2}+Y_{J_1}\otimes Y_{J_2}+Z_{J_1}\otimes Z_{J_2}$.
In the same spirit one finds
\begin{equation}
    \Pi_s = \frac{-\id_A + {\rm S}_A}{2} \otimes \frac{\id_B + {\rm B}_B}{2} \,,
\end{equation}
where ${\rm B}_{J}=X_{J_1}\otimes X_{J_2}-Y_{J_1}\otimes Y_{J_2}+Z_{J_1}\otimes Z_{J_2}$.
Due to these factorization properties, we can see that $\widehat{\TC}_{\mathbb{SP}(2)}^{(2)}$ projects in a basis that can be decomposed as a tensor product of the form  $\{\id, {\rm S}\}_A \otimes \{\id, {\rm S}, {\rm B}\}_B$. Crucially, we remark that the asymmetry of this basis arises from the fact that there is a preferred qubit in our construction of symplectic circuits. 

Then, we recall that it was shown in~\cite{braccia2024computing} that $\widehat{\TC}_{\mathbb{SO}(4)}^{(2)}$ will project on a local tensor product basis of the form $\{\id, {\rm S}, {\rm B}\}_A \otimes \{\id, {\rm S}, {\rm B}\}_B$, meaning that we can describe the full action of $\widehat{\TC}^{(2)}_{\SC}$ as acting on a vector space of dimension $2\times 3^{n-1}$ (which again reveals the special role played by the first qubit). The only thing that remains is to describe the action of $\widehat{\TC}_{\mathbb{SP}(2)}^{(2)}$ and $\widehat{\TC}_{\mathbb{SO}(4)}^{(2)}$ in this space. Given that the description of $\widehat{\TC}_{\mathbb{SO}(4)}^{(2)}$ in the aforementioned basis was already presented in~\cite{braccia2024computing}, we here now only derive that of  $\widehat{\TC}_{\mathbb{SP}(2)}^{(2)}$. In particular, all that we need to do is to compute how this operator acts on all basis states in $\{\id, {\rm S}\} \otimes \{\id, {\rm S}, {\rm B}\}$, which we can compute via the  twirl map of Eq.~\eqref{eq:twirl-2}. A direct calculation reveals that $\widehat{\TC}_{\mathbb{SP}(2)}^{(2)}$ acts on this reduced space as a matrix $\tau$ given by
\begin{equation}
    \tau=\begin{pmatrix}
        1 & 0 & 0 & 0 & 0 & 0 \\
        0 & 1/4 & 3/20 & 1/10 & 1/4 & -3/20 \\
        0 & 3/20 & 1/4 & -1/10 & 3/20 & -1/4 \\
        0 & 3/20 & -3/20 & 3/10 & 3/20 & 3/20 \\
        0 & 3/5 & 0 & 3/5 & 3/5 & 0 \\
        0 & 0 & -3/5 & 3/5 & 0 & 3/5 
    \end{pmatrix}\, .
\end{equation}

\clearpage

\end{document}